\documentclass[onecolumn]{IEEEtran}
\usepackage{amsmath,amsthm,amssymb}
\usepackage{lineno}
\usepackage{geometry}
\usepackage{color}
\usepackage{listings}
\usepackage{algorithm}
\usepackage{algorithmic}
\usepackage{setspace}
\usepackage{multirow}
\usepackage{diagbox} 
\usepackage{tikz}
\usetikzlibrary{trees,arrows}
\usepackage{makecell}
\usepackage{float}
\usepackage{ulem}
\usepackage{textgreek}
\usepackage{bm}

\newtheorem{theorem}{Theorem}[section]
\newtheorem{lemma}{Lemma}[section]

\newtheorem{proposition}{Proposition}[section]
\newtheorem{define}{Definition}[section]
\newtheorem{remark}{Remark}[section]

\newtheorem{example}{Example}[section]
\newtheorem{problem}{\ \ \ Problem}[section]

\def\R{{\mathbb{R}}}

\def\a{{\mathbf{a}}}

\def\C{{\mathcal{C}}}
\def\V{{\mathcal{V}}}

\def\x{{\mathbf{x}}}
\def\b{{\mathbf{b}}}

\lstdefinelanguage{Maple}{
	keywords={if, while, do, else, end, for, from, to,then},
	keywordstyle=\color{blue}\bfseries,
	ndkeywords={class, export, boolean, throw, implements, import, this},
	ndkeywordstyle=\color{darkgray}\bfseries,
	identifierstyle=\color{black},
	sensitive=false,
	comment=[l]{//},
	morecomment=[s]{/*}{*/},
	commentstyle=\color{purple}\ttfamily,
	stringstyle=\color{red}\ttfamily,
	morestring=[b]',
	morestring=[b]"
}

\lstdefinelanguage{SOStools}{
	keywords={syms,sosprogram,monomials,sosineq,sossetobj,sossolve,sosgetsol,sospolyvar},
	keywordstyle=\color{blue}\bfseries,
	ndkeywords={syms,sosprogram,monomials,sosineq,sossetobj,sossolve,sosgetsol},
	ndkeywordstyle=\color{blue}\bfseries,
	identifierstyle=\color{black},
	sensitive=false,
	comment=[l]{//},
	morecomment=[s]{/*}{*/},
	commentstyle=\color{purple}\ttfamily,
	stringstyle=\color{red}\ttfamily,
	morestring=[b]',
	morestring=[b]"
}

\makeatletter
\newenvironment{breakablealgorithm}
{
		\begin{center}
			\refstepcounter{algorithm}
			\hrule height.8pt depth0pt \kern2pt
			\renewcommand{\caption}[2][\relax]{
				{\raggedright\textbf{\ALG@name~\thealgorithm} ##2\par}%
				\ifx\relax##1\relax 
				\addcontentsline{loa}{algorithm}{\protect\numberline{\thealgorithm}##2}%
				\else 
				\addcontentsline{loa}{algorithm}{\protect\numberline{\thealgorithm}##1}%
				\fi
				\kern2pt\hrule\kern2pt
			}
		}{
		\kern2pt\hrule\relax
	\end{center}
}
\makeatother

\ifCLASSINFOpdf
\else
\fi
\begin{document}
%

\title{Proving Information Inequalities by Gaussian Elimination}

%
%
%

\author{
	Laigang~Guo,~
       Raymond~W.~Yeung~
       and~Xiao-Shan~Gao
\thanks{L. Guo is with the School of Mathematical Sciences, Beijing Normal University, Beijing. e-mail: (lgguo@bnu.edu.cn).}
\thanks{R. W. Yeung is with the Institute of Network Coding, The Chinese University of Hong Kong, Hong Kong. e-mail: (whyeung@ie.cuhk.edu.hk).}
\thanks{X.-S. Gao is with the Key Laboratory of Mathematics Mechanization, Chinese Academy of Sciences, Beijing. e-mail: (xgao@mmrc.iss.ac.cn).}
}

\maketitle


\doublespacing

\begin{abstract}
\noindent
The proof of information inequalities and identities under linear constraints on the information measures is an important problem in information theory. 
For this purpose, ITIP and other variant algorithms have been developed and implemented, which are all based on solving a linear program (LP). 
In this paper, we develop a method with symbolic computation. 
Compared with the known methods, our approach can completely avoids the use of linear programming which may cause numerical errors.
Our procedures are also more efficient computationally. 
\end{abstract}

\begin{IEEEkeywords}
Entropy, mutual information, information inequality, information identity, machine proving, ITIP.
\end{IEEEkeywords}

\IEEEpeerreviewmaketitle

\section{Introduction}

In information theory, we may need to prove various information inequalities and identities that involve Shannon's information measures. For example, such information inequalities and identities play a crucial role in establishing the converse of most coding theorems. However, proving an information inequality or identity  
involving more than a few random variables can be highly non-trivial.

To tackle this problem, a framework for linear information inequalities
was introduced in \cite{Yeung1997}. Based on this framework, the problem of 
verifying Shannon-type inequalities can be formulated as a linear program (LP), and a software package based on MATLAB called Information Theoretic Inequality Prover (ITIP) was developed \cite{Yeung.Yan1996}.
Subsequently, different variations of ITIP have been developed. Instead of MATLAB,
Xitip \cite{Pulikkoonattu2006} uses a C-based linear programming solver, and it has  been further developed into its web-based version, oXitip \cite{Rathenakar2020}. 
minitip \cite{Csirmaz2016} is a C-based version of ITIP that adopts a simplified syntax
and has a user-friendly syntax checker.
psitip \cite{Li2020} is a Python library that can verify unconstrained/constrained/existential entropy inequalities. It is a computer algebra system where random variables, expressions, and regions are objects that can be manipulated. 
AITIP \cite{Ho2020} is a cloud-based platform that not only provides analytical proofs for Shannon-type inequalities but also give hints on constructing a smallest counterexample in case the inequality to be verified is not a Shannon-type inequality.

Using the above LP-based approach, to prove an information identity $f = 0$, two LPs need to be solved, one for proving the inequality $f \ge 0$ and the other for proving the inequality $f \le 0$. Roughly speaking, the amount of computation for proving an information identity is twice the amount for proving an information inequality. If the underlying random variables exhibit certain Markov or functional dependence structures, there exist more efficient approaches to proving information identities \cite{Yeung2002}\cite{Chan2019}.

The LP-based approach is in general not computationally efficient because it does not take advantage of the special structure of the underlying LP. To tackle this issue, we developed in \cite{Guo-Yeung-Gao2023} a set of algorithms that can be implemented by symbolic computation. Based on these algorithms, we devised procedures for reducing the original LP to the minimal size, which can be solved easily. These procedures are computationally more efficient than solving the original LP directly. 
%
In this paper, we develop a different symbolic approach which not only make the reductions from the original LP to the minimal size more efficient, but also in many cases can prove the information inequality without solving any LP. 

The specific contributions of this paper are:

1) We develop a heuristic method to prove an information inequality. This heuristic method does not prove an information inequality by directly solving the associated LP, but rather expedites the proof process through polynomial reduction (Gaussian elimination).

2) This heuristic method may not succeed in proving the inequality. If it does not succeed, it can simplify the original LP into a smaller LP.



3) We give several examples that verify the effectiveness of our method.

\section{Information inequality preliminaries}
In this section, we present some basic results related to information inequalities and their verification. For a comprehensive discussion on the topic, we refer the reader to 
\cite[Chs.~13-15]{Yeung-Li2021,Yeung2008}.

It is well known that all Shannon's information measures, namely entropy, conditional entropy, mutual information, and conditional mutual information are always nonnegative. The nonnegativity of all Shannon's information measures forms a set of
inequalities called the {\it basic inequalities}. The set of basic inequalities, however, is not minimal in the sense that some basic inequalities are implied by the others. For example,
$$H(X|Y)\geq0\ {\rm and}\ I(X;Y)\geq0,$$
which are both basic inequalities involving random variables $X$ and $Y$, imply
$$H(X)=H(X|Y)+I(X;Y)\geq0,$$
again a basic inequality involving $X$ and $Y$.

Throughout this paper, all random variables are discrete. Unless otherwise specified, all 
information expressions involve some or all of the random variables $X_1,X_2, \ldots,X_n$. The value of $n$ will be specified when necessary.  
Denote the set $\{1,2,\ldots,n\}$ by $\mathcal{N}_n$,  the set $\{0,1,2,\ldots\}$ by $\mathbb{N}_{\ge0}$ and the set $\{1,2,\ldots\}$ by $\mathbb{N}_{>0}$.

\begin{theorem}{\rm \cite{Yeung1997}}
	\label{element-ine}
	Any Shannon's information measure can be expressed as a conic combination of the following two elemental forms of
	Shannon's information measures:
	
	i) $H(X_i|X_{\mathcal{N}_n-\{i\}})$
	
	ii) $I(X_i;X_j|X_K)$, where $i\neq j$ and $K\subseteq \mathcal{N}_n-\{i,j\}$.
	
\end{theorem}

The nonnegativity of the two elemental forms of Shannon’s information measures forms a proper but equivalent subset of the set of basic inequalities. The inequalities in this smaller set are called the {\it elemental inequalities}.
In \cite{Yeung1997}, the minimality of the elemental inequalities is also proved.
%
The total number of elemental inequalities is equal to 
\begin{equation*}
	m=n+\sum\limits_{r=0}^{n-2}\left(
	\begin{array}{c}
		n\\
		r\\
	\end{array}
	\right)
	\left(
	\begin{array}{c}
		n-r\\
		2\\
	\end{array}
	\right)
	=n+
	\left(
	\begin{array}{c}
		n\\
		2\\
	\end{array}
	\right)  2^{n-2}  .
	\label{m}
\end{equation*}

In this paper, inequalities (identities) involving only Shannon's information measures
are referred to as information inequalities (identities). 
The elemental inequalities are called {\it unconstrained} information inequalities
because they hold for all joint distributions of the random variables.
In information theory, we very often deal with information inequalities (identities) that hold  under certain constraints on the joint distribution of the random variables. These are called {\it constrained} information inequalities (identities), and the associated constraints are usually  expressible as linear constraints on Shannon's information measures. We will 
confine our discussion to constrained inequalities of this type.

\begin{example}
	\def\ra{\rightarrow}
	The celebrated data processing theorem asserts that for any four random variables 
	$X$, $Y$, $Z$ and $T$, if $X \rightarrow Y \rightarrow Z \rightarrow T$ 
	forms a Markov chain, then
	$I(X;T) \le I(Y;Z)$. Here, $I(X;T) \le I(Y;Z)$ is a constrained information inequality under the constraint $X \rightarrow Y \rightarrow Z \rightarrow T$, which is equivalent to 
	\[
	\left\{
	\begin{array}{rcl}
		I(X;Z|Y) & = & 0 \\
		I(X,Y;T|Z) & = & 0 ,
	\end{array}
	\right.
	\]
	or 
	\[
	I(X;Z|Y) + I(X,Y;T|Z) = 0
	\]
	owing to the nonnegativity of conditional mutual information. Either way, the Markov
	chain can be expressed as a set of linear constraint(s) on Shannon's information measures.	
\end{example}

Information inequalities (unconstrained or constrained) that are implied by the 
basic inequalities are called {\it Shannon-type} inequalities. Most of the information 
inequalities that are known belong to this type. However, {\it non-Shannon-type} 
inequalities do exist, e.g., \cite{Zhang1998}. See \cite[Ch.~15]{Yeung2008} for a discussion.

Shannon's information measures, with conditional mutual information 
being the general form, can be expressed as a linear combination of joint entropies by means of following identity:
$$ I(X_G;X_{G'}|X_{G''})=H(X_G,X_{G''})+H(X_{G',G''})-H(X_G,X_{G'},X_{G''})-H(X_{G''}). $$
where $G,G',G''\subseteq \mathcal{N}_n$. 
For the random variables $X_1, X_2, \ldots, X_n$, there are a total of $2^n-1$ joint entropies.
By regarding the joint entropies as variables, the basic (elemental) inequalities
become linear inequality constraints in $\R^{2^n-1}$. By the same token, 
the linear equality constraints on 
Shannon's information measures imposed by the problem under discussion become
linear equality constraints in $\R^{2^n-1}$. This way, the problem of verifying 
a (linear) Shannon-type inequality can be formulated as a linear program (LP),
which is described next.

Let $\bf{h}$ be the column $(2^n-1)$-vector of the joint entropies 
of $X_1, X_2, \ldots, X_n$. The set of elemental inequalities can be written as 
${\bf G h} \ge 0$, where $\bf{G}$ is an $m \times (2^n-1)$ matrix and $ {\bf G h} \ge 0$ means
all the components of ${\bf G h}$ are nonnegative. Likewise, the constraints on 
the joint entropies can be written as ${\bf Q h} = 0$. When there is no constraint on the joint entropies, $\bf{Q}$ is assumed to contain zero rows. The following theorem enables
a Shannon-type inequality to be verified by solving an LP.

\begin{theorem} {\rm \cite{Yeung1997}}
	${\bf b}^\top {\bf h} \ge 0$ is a Shannon-type inequality under the constraint $ {\bf Q h} = 0$ 
	if and only if the minimum of the problem
	\begin{center}
		Minimize ${\bf b}^\top {\bf h}$, subject to ${\bf G h} \ge 0$ and ${\bf Q h} = 0$
	\end{center}
	is zero.  
	\label{LP-S}
\end{theorem}

\section{Algorithms for homogeneous linear inequlities}
\label{Algorithms}

In this section, we will develop new algorithms for proving information inequalities and identities. We will start by discussing some notions pertaining to linear inequality sets and linear equality sets. Then we will state some related properties that are necessary for developing these algorithms. For details, one can refer to \cite{Schrijver2003,Ben-Tal-Nemirovski2022}.

Let $\mathbf{x}=(x_1,x_2,\ldots,x_n)^T$, and let $\R_h[\x]$ be the set of all homogeneous linear polynomials in $\x$ with real coefficients.
In this paper, unless otherwise specified, we assume that all inequality sets have the form $S_f=\{f_i\geq0,i\in\mathcal{N}_m\}$, with $f_i\neq0$ and $f_i\in\mathbb{R}_h[\mathbf{x}]$, and all  equality sets have the form $E_{\tilde{f}}=\{\tilde{f}_i=0,i\in\mathcal{N}_{\widetilde{m}}\}$ with $\tilde{f}_{i}\neq0$ and $\tilde{f}_i\in\mathbb{R}_h[\mathbf{x}]$.

For a given set of polynomials $P_f=\{f_i,i\in\mathcal{N}_m\}$ and the corresponding set of inequalities $S_f=\{f_i\geq0,i\in\mathcal{N}_m\}$, and a given set of polynomials $P_{\tilde{f}}=\{\tilde{f}_i,i\in\mathcal{N}_{\widetilde{m}}\}$ and the corresponding set of equalities $E_{\tilde{f}}=\{\tilde{f}_i=0,i\in\mathcal{N}_{\widetilde{m}}\}$, where $f_i$ and $\tilde{f}_i$ are polynomials in $\mathbf{x}$, we write
	$S_f=\mathcal{R}(P_f)$, $P_f=\mathcal{R}^{-1}(S_f)$,
	$E_{\tilde{f}}=\widetilde{\mathcal{R}}(P_{\tilde{f}})$ and $P_{\tilde{f}}=\widetilde{\mathcal{R}}^{-1}(E_{\tilde{f}})$.
	%
	
	\begin{define}\label{subset}
		Let $S_f=\{f_i\geq0,i\in\mathcal{N}_m\}$ and $S_{f'}=\{f'_i\geq0,i\in\mathcal{N}_{m'}\}$ be two inequality sets, and $E_{\tilde{f}}$ and $E_{\tilde{f}'}$ be two equality sets.  We write $S_{f'}\subseteq S_{f}$ if $\mathcal{R}^{-1}(S_{f'})\subseteq \mathcal{R}^{-1}(S_{f})$, and $E_{\tilde{f}'}\subseteq E_{\tilde{f}}$ if $\widetilde{\mathcal{R}}^{-1}(E_{\tilde{f}'})\subseteq \widetilde{\mathcal{R}}^{-1}(E_{\tilde{f}})$. 
		Furthermore, we write $(f_i\geq0) \in S_f$ to mean that the inequality $f_i\geq0$ is in $S_f$.
	\end{define}

	
	
	\begin{define}
		Let $\mathbb{R}_{>0}$ and $\mathbb{R}_{\geq0}$ be the sets of positive and nonnegative real numbers, respectively.
		A linear polynomial $F$ in $\mathbf{x}$ is called a {\it positive (nonnegative) linear combination} of polynomials $f_j$ in $\mathbf{x}$, $j=1,\ldots,m$, if $F=\sum_{j=1}^{m}r_jf_j$ with $r_j\in$ $\mathbb{R}_{>0}$ $(r_j\in$ $\mathbb{R}_{\geq0})$.
		%
		A nonnegative linear combination is also called a {\it conic} combination.
	\end{define}

\begin{define}
	The inequalities $f_1\geq0,f_2\geq0,\ldots,f_m\geq0$ imply the inequality $f\geq0$ if the following holds:
	\begin{center}
	For all $\mathbf{x}\in\mathbb{R}^n$, $\mathbf{x}$ satisfying $f_1\geq0,f_2\geq0,\ldots,f_m\geq0$ implies $\mathbf{x}$ satisfies $f\geq0$.
	\end{center}
	
\end{define}

\begin{define}\label{def-redundant}
	Given a set of inequalities $S_f=\{f_i\geq0,i\in\mathcal{N}_m\}$, for $i\in\mathcal{N}_m$, $f_i\geq0$ is called a redundant inequality if $f_i\geq0$ is implied by the inequalities $f_j\geq0$, where $j\in \mathcal{N}_m\backslash\{i\}$.
	
\end{define}


\begin{define}\label{def-purein}
	Let $S_f=\{f_i(\mathbf{x})\geq0,i\in\mathcal{N}_{m}\}$ be an inequality set.
	If $f_{k}(\mathbf{x})=0$ for all solutions $\mathbf{x}$ of $S_f$,
	then $f_{k}(\mathbf{x})=0$ is called an {\it implied equality} of $S_f$. The inequality set $S_f$ is called a {\it pure inequality set} if $S_f$ has no implied equalities.
\end{define}

	\begin{lemma}{\rm \cite{Guo-Yeung-Gao2023}}
		\label{equa1}
		Let $S_f=\{f_i(\mathbf{x})\geq0,i\in\mathcal{N}_{m}\}$ be an inequality set.
		Then $f_k=0$ is an implied equality of $S_f$ if and only if 
		\begin{equation}
			f_{k}(\mathbf{x})\equiv\sum\limits_{i=1,i\neq k}^{m}p_{i}f_{i}(\mathbf{x}),
			\label{qqoinva}
		\end{equation}
		where $p_i \le 0$ for all $i \in {\cal N}_m \backslash \{k\}$.
		
	\end{lemma}

\begin{lemma}{\rm \cite{Ben-Tal-Nemirovski2022}}
	\label{GG}
	Given $h_1,\ldots,h_m,{h}_0\in\mathbb{R}_h[\mathbf{x}]$,  $h_1\ge0,...,h_m\ge0$ imply ${h}_0\ge0$ if and only if ${h}_0$ is a conic combination of $h_1,\ldots,h_m$.
\end{lemma}

\begin{define}
Let $f\in\R_h[\x]$ and $x_1\succ x_2\succ\cdots x_n$ be a fixed variable order. 
The variable set of $f$, denoted by $V(f)$, is the set containing all the variables of $f$.
The variable sequence of $f$, denoted by $\V(f)$, is the sequence containing all the variables of $f$ in the given order.
The coefficient sequence of $f$, denoted by $\C(f)$, is the sequence containing the coefficients corresponding to the variables in $\V(f)$.
We adopt the convention that $\C(f)=[0]$ and $V(f)=\emptyset$ for $f\equiv0$.
\end{define}

\begin{define}
Let $F$ be a polynomial in $\mathbf{x}$, then $|F|$ denotes the number of variables in $F$.
\end{define}


\begin{define}
Let $P_f=\{f_i,\ i\in\mathcal{N}_m\}$, where $f_i\in\R_h[\x]$. The variable set of $P_f$, denoted by $V(P_f)$, is the set containing all the variables of $f_i$'s, i.e., 
$V(P_f)=\cup_{i\in\mathcal{N}_m}V(f_i)$.
\end{define}

\begin{example}
Let $P_f=\{f_1,f_2\}$, where $f_1=x_1+x_2,\ f_2=x_1-x_3$. Then,  we have 

$$
V(f_1)=\{x_1,x_2\},\ \V(f_1)=[x_1,x_2],\ \C(f_1)=[1,1],\  V(f_2)=\{x_2,x_3\},\ {\rm and}\ V(P_f)=\{x_1,x_2,x_3\}.
$$

\end{example}

Observe that for any polynomial $f(\x)$, the following equality holds:

$$\{\x: f(\x)\ge0\}={\rm Proj}_{\x}\{(\x,a): f(\x)-a=0, a\ge0\}.$$

\noindent Note that on the RHS, a new variable $a$ is introduced. Motivated by this observation, in the sequel we will say that an inequality $f(\x)\ge0$ is equivalent to the semi-algebraic set $\{f(\x)-a=0,\ a\ge0\}$.
Also, $\{f_i(\x)\ge0, i\in\mathcal{N}_m\}$ is equivalent to $\{f_i(\x)-a_i=0,\ a_i\ge0, i\in\mathcal{N}_m\}$.

The following proposition is well known (see for example \cite[Chapter 1]{Lay2016}).
	\begin{proposition}\label{rsb-def}
		Under the variable order $x_1\succ x_2\succ \cdots\succ x_n$,
		the linear equation system  $E_{\tilde{f}}=\{\tilde{f}_i=0,\ i\in\mathcal{N}_{\widetilde{m}}\}$
		can be reduced by the Gauss-Jordan elimination to the unique form 
		\begin{equation}\label{Sta-Basis}
			\widetilde{E}=\{x_{k_i}-U_i=0,i\in \mathcal{N}_{\widetilde{n}}\},
		\end{equation}
		where $\widetilde{n}$ is the rank of the linear system $E_{\tilde{f}}$, $k_1<k_2<\cdots<k_{\widetilde{n}}$, $x_{k_i}$ is the leading term of $x_{k_i}-U_i$, and $U_i$ is a linear function in 
		$\{x_j, \hbox{ for }k_i<j\le n,\ \ j\neq k_{l},\ i<l\le \tilde{n}\}$, with $k_{\tilde{n}+1}=n+1$ by convention. 

		Among $x_1,x_2,\ldots,x_n$, the variable $x_{k_i},\ i\in\mathcal{N}_{\tilde{n}}$ are called the pivot variables, and the rest are called the free variables. 
	\end{proposition}

We call the equality set $\widetilde{E}$ the reduced row echelon form (RREF) of $E_{\tilde{f}}$. Likewise, we call the polynomial set $\widetilde{\mathcal{R}}^{-1}(\widetilde{E})$ the RREF of  $\widetilde{\mathcal{R}}^{-1}(E_{\tilde{f}})$.
We say applying the Gauss-Jordan elimination to $\widetilde{\mathcal{R}}^{-1}(E_{\tilde{f}})$ to mean finding $\widetilde{\mathcal{R}}^{-1}(\widetilde{E})$ by Proposition~\ref{rsb-def}.

\begin{define}
Let $H=\{h_i,i\in\mathcal{N}_{m}\}$ be a set of polynomials, where $h_i\in\R_h[\b]$ and $\b=(x_1,\ldots,x_n,a_1,\ldots,a_m)^T$. 
Under the variable order $x_1\succ\cdots\succ x_n\succ a_1\succ\cdots\succ a_m$, we can obtain the RREF of $H$, denoted by $\widetilde{H}$. 
Let $\widetilde{H}=H_1\cup H_2$, where 

$V(h)\cap\{x_1,x_2,\ldots,x_n\}\neq\emptyset$ for every $h\in H_1$,

$V(h)\cap\{x_1,x_2,\ldots,x_n\}=\emptyset$ and $V(h)\cap\{a_1,a_2,\ldots,a_m\}\neq\emptyset$ for every $h\in H_2$.

\noindent $H_1$ is called the partial RREF of $H$ in $\x$ and $\a$, and $H_2$ is called the partial RREF of $H$ in $\a$.
\end{define}

\begin{algorithm}[H]
	\caption{Dimension Reduction}
	\label{division-algorithm}
	\begin{algorithmic}[1]
		
		\REQUIRE $S_f$, $E_{\tilde{f}}$.
		\ENSURE The remainder set $R_f$.
		\STATE Compute $\widetilde{E}$ for $E_{\tilde{f}}$ by Proposition \ref{rsb-def}.
		\STATE Substitute $x_{k_i}$ by $U_i$ in $\mathcal{R}^{-1}(S_f)$ to obtain the set $R$.
		\STATE Let $R_f=R\backslash\{0\}$. 
		
		\RETURN ${\cal R}(R_f)$.
	\end{algorithmic}
\end{algorithm}

We say reducing $S_f$ by $E_{\tilde{f}}$ to mean using Algorithm~\ref{division-algorithm} to find ${\cal R}(R_f)$.
We also say reducing $P_f$ by $E_{\tilde{f}}$ to mean using Algorithm~\ref{division-algorithm} to find $R_f$, called {\it the remainder set} (or remainder if $R_f$ is a singleton).

\begin{define}
	Let $E_{\widetilde{f}}=\{\widetilde{f}_i=0,i\in\mathcal{N}_{\widetilde{m}}\}$ and $E_{f'}=\{f'_i=0,i\in\mathcal{N}_{m'}\}$ be two equality sets, where $\tilde{f}_i,\ f'_i\in\mathbb{R}_h[\mathbf{x}]$.
	If the solution sets of  $E_{f'}$ and $E_{\widetilde{f}}$ are the same,  then we say that $E_{\widetilde{f}}$ and $E_{f'}$ are equivalent.\footnote{With a slight of abuse of terminology, the solution set of $E_{\tilde{f}}$ refers to the set $\{(x_1,x_2,\ldots,x_n)\in\mathbb{R}^n:\tilde{f}_i=0,\ i\in\mathcal{N}_{\widetilde{m}}\}$.}
\end{define}

\begin{define}\label{transf.}
	Let $h_i\in\R_h[\bm{a}]$, $i=1,2$, where $\a=(a_1,\ldots,a_m)^T$ and let $E_{\tilde{f}}=\{\tilde{f}_i=0,i\in\mathcal{N}_{\widetilde{m}}\}$ be an equality set,
	where $\tilde{f}_i\in\R_h[\bm{a}]$ for all $i\in\mathcal{N}_{\widetilde{m}}$.
	 We say $h_1$ can be transformed to $h_2$ by $E_{\tilde{f}}$ 
	 if $h_1\equiv h_2+h_3$, where $h_3\equiv\sum\limits_{i=1}^{m'}{q}_if'_i$, $q_i\in\R$ and $E_{f'}=\{f'_i=0,i\in\mathcal{N}_{m'}\}$ is an equivalent set of $E_{\tilde{f}}$.
	
\end{define}

Let $F_0\in\R_h[\x]$ and $S_f=\{f_i\ge0,i\in\mathcal{N}_m\}$, where  $f_i\in\R_h[\x]$. In the rest of this section, we discuss how to solve the following problem.

\begin{problem}\label{problem-p4}
	Prove $F_0\ge0$ subject to $S_f$.
\end{problem}

We first give a method implemented by the following algorithm for reducing Problem \ref{problem-p4} to another LP.
\\[0.2cm]
\begin{breakablealgorithm}
	\caption{LP reduction Algorithm}
	\label{Algorithm-transf.}
	\begin{algorithmic}[1]
		\REQUIRE \ Problem \ref{problem-p4}
		\ENSURE \  A reduced LP.
		\STATE Let $G_i=f_i-a_i$,  $i\in \mathcal{N}_m$, where $a_i$'s are assumed to satisfy $a_i\ge0,\ i\in \mathcal{N}_m$.
		\STATE Fix the variable order $x_1\succ x_2\succ\cdots\succ x_n\succ a_1\succ\cdots\succ a_m$.
		\STATE Apply the Gauss-Jordan elimination to $\{G_i,i\in\mathcal{N}_{m}\}$ and obtain the RREF. 
		\STATE Let $J_0$ be the partial RREF of $\{G_i,i\in\mathcal{N}_{m}\}$ in $\x$ and $\a$, and $J_1$ be the partial RREF of $\{G_i,i\in\mathcal{N}_{m}\}$ in $\a$.
		\STATE Reduce $F_0$ by $J_0$ to obtain $F$. 
	    \STATE The Problem \ref{problem-p4} is equivalent to 
		\begin{problem}\label{problem-p5}
			Prove $F\ge0$ subject to $\widetilde{\mathcal{R}}(J_1)$ and $a_i\ge0,\ i\in\mathcal{N}_{m}$.
		\end{problem}
		\RETURN Problem \ref{problem-p5}.
	\end{algorithmic}
\end{breakablealgorithm}
\ 
\\[0.1cm]

\begin{remark}
In Algorithm \ref{Algorithm-transf.}, if Problem \ref{problem-p4} can be solved, then $F$ needs to satisfy $V(F)\cap\{x_1,\ldots,x_n\}=\emptyset$.
If there exist $x_i\in V(F)$, then $x_i$ is a free variable in Problem \ref{problem-p5}, and Problem \ref{problem-p5} cannot be solved. Thus Problem \ref{problem-p4} cannot be solved. For example, we consider the problem

\textbf{P1}: Prove $x_1+x_3\ge0$ subject to $x_1\ge0$ and $x_2\ge0$. 

\noindent Running Algorithm 2, the above problem becomes

\textbf{P2}: Prove $a_1+x_3\ge0$ subject to $a_1\ge0$.

\noindent Obviously, \textbf{P2} cannot be proved since $x_3$ is a free variable.

\end{remark}

Let $\bm{a}=(a_1,\ldots,a_{m})^T$, $F\in\R_h[\bm{a}]$, $f_i\in\R_h[\bm{a}]$ for $i\in\mathcal{N}_{\widetilde{m}}$,
$S_{a}=\{a_i\ge0,\ i\in\mathcal{N}_{m}\}$, and $E_{a}=\{f_i=0,\ i\in\mathcal{N}_{\widetilde{m}}\}$. 
Based on the discussion above, we only need to consider the case that $F$ satisfies $V(F)\cap \{x_1,\ldots,x_n\}=\emptyset$.

To facilitate the discussion, we restate Problem \ref{problem-p5} in a general form:

\begin{problem}\label{problem-op}
	Prove $F\ge0$ subject to $E_{a}$ and $S_{a}$.
\end{problem}

\noindent We say that a problem as given in Problem \ref{problem-op} is ``solvable" if $F\ge0$ is implied by $E_a$ and $S_a$.

\begin{theorem}\label{Combofalpha}
Problem \ref{problem-op} is solvable if and only if $F$ can be transformed into a conic combination of $a_i, i\in\mathcal{N}_m$ by $E_{a}$.
\end{theorem}
\begin{proof}
The sufficiency is obvious. We only need to prove the necessity.

Assume that Problem \ref{problem-op} is solvable. By Proposition \ref{rsb-def}, we compute the RREF of $E_{a}$, denoted by $\widetilde{E}=\{a_{k_i}-U_i=0,\ i\in\mathcal{N}_{\widetilde{n}}\}$, and substitute $a_{k_i}$ in $F$ by $U_i$ to obtain $F_1$. Then we see that

\begin{equation}\label{gau-1}
F\equiv F_1+\sum\limits_{i=1}^{\widetilde{n}}{q}_i(a_{k_i}-U_i),\ q_i\in\R.
\end{equation}

\noindent In other words, $F$ can be transformed to $F_1$ by $E_{a}$.
Then we substitute $a_{k_i}$ in $S_{a}$ by $U_i$ to obtain 
$S_o=\{g_i\ge0,\ i\in\mathcal{N}_{m}\}$, 
where $g_{k_i}=U_i$ for $i\in\mathcal{N}_{\widetilde{n}}$
and $g_j=a_j$ for $j\in\mathcal{N}_{m}\backslash\{k_i, i\in\mathcal{N}_{\widetilde{n}}\}$. 

Now Problem \ref{problem-op} is equivalent to

\begin{problem}\label{problem-p1}
Prove $F_1\ge0$ subject to $S_o$.
\end{problem}

\noindent By Lemma \ref{GG}, Problem \ref{problem-p1} is solvable if and only if $F_1$ is a conic combination of $g_i,\ i\in\mathcal{N}_m$.
Suppose $F_1\equiv\sum_{i=1}^{m}{p_ig_i}$ with $p_i\in\R_{\ge0}$. Then
\begin{equation}\label{eq-comb}
	\begin{array}{ll}
F_1\equiv\sum\limits_{j=1}^{m}{p_jg_j}\\[0.2cm]
\ \ \ \ \equiv\sum\limits_{i=1}^{\widetilde{n}}p_{k_i} U_i+\sum\limits_{j\in\mathcal{N}_{m}\backslash\{k_i, i\in\mathcal{N}_{\widetilde{n}}\}}p_ja_j\\[0.2cm]
\ \ \ \ \equiv\sum\limits_{i=1}^{\widetilde{n}}p_{k_i} a_{k_i}-\sum\limits_{i=1}^{\widetilde{n}}p_{k_i} (a_{k_i}-U_i)+\sum\limits_{j\in\mathcal{N}_{m}\backslash\{k_i, i\in\mathcal{N}_{\widetilde{n}}\}}p_ja_j\\[0.2cm]
\ \ \ \ \equiv\sum\limits_{i=1}^{m}p_ia_i-\sum\limits_{i=1}^{\widetilde{n}}p_{k_i} (a_{k_i}-U_i).\\[0.2cm]
\ \ \ \ =\sum\limits_{i=1}^{m}p_ia_i,
\end{array}\end{equation}
where the last step follows from the constraints in $\widetilde{E}$.

So, $F_1$ can be expressed as a conic combination of $a_i$'s. 
Then, by \eqref{gau-1} and the second last line above, we obtain
\begin{equation}\label{conic-c}
	\begin{array}{ll}
F\equiv F_1+\sum\limits_{i=1}^{\widetilde{n}}{q}_i(a_{k_i}-U_i)\\
\ \ \ \equiv\sum\limits_{i=1}^{m}p_ia_i-\sum\limits_{i=1}^{\widetilde{n}}p_{k_i} (a_{k_i}-U_i)+\sum\limits_{i=1}^{\widetilde{n}}{q}_i(a_{k_i}-U_i)\\
\ \ \ \equiv\sum\limits_{i=1}^{m}p_ia_i+\sum\limits_{i=1}^{\widetilde{n}}({q}_i-p_{k_i})(a_{k_i}-U_i)
\end{array}\end{equation}
where $p_i\in\R_{\ge0}$, $p_{k_i}\in\R_{\ge0}$ and $q_i\in\R$.

From Definition \ref{transf.}, we see that
$F_1$ is a conic combination of $g_i,\ i\in\mathcal{N}_m$ if and only if $F$ can be transformed into a conic combination of $a_i,\ i\in\mathcal{N}_m$ by  $E_{a}$.
Hence, following the discussion in the foregoing, the theorem is proved.
\end{proof}

\begin{define}
	Let $E_a=\{f_i=0,\ i\in\mathcal{N}_{\widetilde{m}}\}$, where $f_i$ is a polynomial in $\mathbf{a}$, be an equality set. We say eliminating a variable $a_i$ from $E_a$ to mean solving for $a_i$ in some $f_i=0$ with $a_i\in V(f_i)$ to obtain $a_i=A_i$ and then substituting $a_i=A_i$ into $E_a$ to obtain $E_A=subs(a_i=A_i,E_a)\backslash\{0=0\}$.
	
	Let $F$ be a polynomial in $\mathbf{a}$. We say eliminating $a_i$ from $F$ by $E_a$ to mean eliminating $a_i$ from $E_a$ to obtain $a_i=A_i$ and $E_A$, and then substituting $a_i=A_i$ into $F$ to obtain $F_1=subs(a_i=A_i,F)$.
\end{define}

The notions of redundant inequality and implied equality in Definitions \ref{def-redundant} and \ref{def-purein}, respectively can be applied in the more general setting in Problem \ref{problem-op}. Specifically, $a_i=0,\ i\in\mathcal{N}_m$ is an implied equality if $-a_i\ge0$ is provable subject to $E_a$ and $S_a$.
Also, by eliminating $a_i$ for some $i\in\mathcal{N}_m$ from $E_a$ to obtain $a_i=A_i$ and $E_A$, $a_i\ge0$ is a redundant inequality if $A_i\ge0$ is provable subject to $E_A$ and $S_a\backslash\{a_i\ge0\}$.
	
\begin{example}
Let $S_a=\{a_i\ge0,\ i\in\mathcal{N}_5\}$ and $E_a=\{f_1=0, f_2=0\}$, where $f_1=a_1+a_2$ and $f_2=a_3-a_4-a_5$. 
Using $f_1=0$, $a_1\ge0$ and $a_2\ge0$,  we can obtain that $-a_1\ge0$ and $-a_2\ge0$.
Thus $a_1=0$ and $a_2=0$ are implied equalities.

By eliminating $a_3$ from $E_a$, we obtain $a_3=a_4+a_5$ and $E_A=\{a_1+a_2=0\}$. Since $a_4+a_5\ge0$ is obviously provable subject to $E_A$ and $S_a\backslash\{a_3\ge0\}$, we have that $a_3\ge0$ is a redundant inequality.
\end{example}


\begin{define}
Let	$f$ be a polynomial in $\mathbf{a}=\{a_1,a_2,\ldots,a_m\}$.
Let $\widetilde{m}\le m$ and $j_1,j_2,\ldots,j_{\widetilde{m}}$ be distinct elements of $\{1,2,\ldots,m\}$.
	If  $f=\sum_{i=1}^{\bar{m}}p_ia_{j_i}$ or $f=-\sum_{i=1}^{\bar{m}}p_ia_{j_i}$ with $p_{i}>0$, then $f$ is called a Type I linear combination of $a_{j_i}$. 
	If $f=\sum_{i=1}^{\bar{m}-1}p_ia_{j_i}-p_{\bar{m}}a_{j_{\bar{m}}}$ or $f=-\sum_{i=1}^{\bar{m}-1}p_ia_{j_i}+p_{\bar{m}}a_{j_{\bar{m}}}$ with $p_{i}>0$, then $f$ is called a Type II linear combination of $a_{j_i}$, and let $single(f)=a_{j_{\bar{m}}}$. 
\end{define}

\begin{define}
In Problem \ref{problem-op}, if $(f=0)\in E_a$ and
\begin{itemize}
\item[1)] if $f$ is Type I, then $a_{i}=0$ for $a_{i}\in V(f)$ are called trivially implied equalities;
\item[2)] if $f$ is Type II, then $single(f)\ge0$ is called a trivially redundant inequality.
\end{itemize}
\end{define}

\begin{example}
Let $E_a=\{f_i=0,\ i\in\mathcal{N}_4\}$, where $f_1=a_1+a_2$, $f_2=-a_1-a_2$, $f_3=a_4-a_5-a_6$, and $f_4=a_7+a_8-2a_9$. Then
$f_1$ and $f_2$ are Type I, $f_3$ and $f_4$ are Type II, $single(f_3)=a_4$, and $single(f_4)=a_9$.
It can readily be checked that $a_1=0$ and $a_2=0$ are trivially implied equalities, and $a_4\ge0$ and $a_9\ge0$ are trivially redundant inequalities.
\end{example}

In the rest of the paper, we denote the $i$th element of a sequence $B$ by $B[i]$.
We also denote the $i$th element of a set $S$ by $S[i]$, where the elements in $S$ are assumed to be sorted in lexicographic order.
For example, $x_1+2x_2\succ x_2+x_5$ and $x_3+x_5\succ x_3+x_6$.

Now we develop an algorithm to remove all trivially implied equalities and trivially redundant inequalities in Problem~\ref{problem-op}. 
To facilitate the discussion, we use $subs(\cdot,\cdot)$ to denote eleminating one or more variables from a set of polynomials by substitution.
The use of this notation will be illustrated in Example \ref{emIII.4}.
 \ 
\\
\begin{breakablealgorithm}
	\caption{Preprocessing Problem \ref{problem-op}}
	\label{QuickIMP-RED}
	\begin{algorithmic}[1]
		\REQUIRE \ Problem \ref{problem-op}.
		\ENSURE \  A reduced LP for Problem \ref{problem-op}.
		\\[0.4cm]
		\STATE Let $E_1:=\widetilde{\mathcal{R}}^{-1}(E_{a})$, $S_{1}:=\mathcal{R}^{-1}(S_{a})$, $F_1:=F$, $i_1:=1$.
	    \WHILE{$i_1=1$}
	    \STATE Let $i_1:=0.$
		\FOR{$i$ from $1$ to $|E_1|$}
		\STATE Let $f:=E_1[i]$.
		\IF{$f$ is Type I}
		\STATE // In this case, all equalities in $\widetilde{\mathcal{R}}(V(f))$ are trivially implied equalities.
		\STATE $E_1:=subs(\widetilde{\mathcal{R}}(V(f)),E_1)\backslash\{0\}$.
		\STATE $S_{1}:=S_{1}\backslash V(f)$.
		\STATE $F_1:=subs(\widetilde{\mathcal{R}}(V(f)),F_1)$.
		\STATE $i_1:=1$.
		\ENDIF
		\IF{$f$ is Type II}
		\STATE // In this case, the inequality $single(f)\ge0$ is a trivially redundant inequality.
		\STATE $E_1:=subs(single(f)=solve(f,single(f)),E_1)\backslash\{0\}$.
		\STATE $S_{1}:=S_{1}\backslash\{single(f)\}$.
		\STATE $F_1:=subs(single(f)=solve(f,single(f)),F_1)$.
		\STATE $i_1:=1$.
		\ENDIF
		\ENDFOR
		\ENDWHILE
		\RETURN  
		A reduced LP:
		\begin{problem}\label{problem-op1}
			Prove $F_1\ge0$ subject to $\widetilde{\mathcal{R}}(E_1)$ and $\mathcal{R}(S_1)$.
		\end{problem}		
	\end{algorithmic}
\end{breakablealgorithm}
 \ 
 \\[0.005cm]
 \begin{example}\label{emIII.4}
 	We want to prove $F=a_1+2a_2-a_3\ge0$ subject to $E_a=\{a_1+a_2-a_3-a_4-a_5=0,\ a_1+a_4=0\}$ and $S_a=\{a_i\ge0,\ i\in\mathcal{N}_5\}$.
Following Algorithm \ref{QuickIMP-RED}, we give the steps in detail.

Step 1. $F_1=a_1+2a_2-a_3$, $E_1=\{a_1+a_2-a_3-a_4-a_5,\ a_1+a_4\}$, $S_1=\{a_1,a_2,a_3,a_4,a_5\}$.

Step 2. Since $a_1+a_4$ is Type I, we obtain $a_1=0$ and $a_4=0$ from $a_1+a_4=0$, $a_1\ge0$, and $a_4\ge0$.

Step 3. Obtain $E_1:=subs(a_1=0, a_4=0, E_1\})\backslash\{0\}=\{a_2-a_3-a_5\}$, $F_1:=subs(a_1=0, a_4=0, F_1)=2a_2-a_3$, 

and $S_1:=S_1\backslash\{a_1,a_4\}=\{a_2,a_3,a_5\}$.

Step 4. Now $a_2-a_3-a_5$ is Type II. Then solve $a_2$ from $a_2-a_3-a_5$ to obtain $a_2=a_3+a_5$. 

Step 5. Obtain $F_1:=subs(a_2=a_3+a_5,F_1)=a_3+2a_5$, $E_1:=subs(a_2=a_3+a_5,E_1)\backslash\{0\}=\emptyset$, and $S_1:=S_1\backslash\{a_2\}=\{a_3,a_5\}$.

Now the reduced problem is to prove $a_3+2a_5\ge0$ subject to $a_3\ge0$ and $a_5\ge0$, which is obviously solvable.
 \end{example}

Algorithm \ref{QuickIMP-RED} removes all the trivially implied equalities and trivially redundant inequalities from Problem \ref{problem-op}. In Appendix A, we will develop two enhancements of Algorithm \ref{QuickIMP-RED}: Algorithm \ref{FindImp} for removing all implied equalities and Algorithm \ref{Findredundant} for removing all redundant inequalities.


Toward solving Problem \ref{problem-op}, we first apply Algorithm \ref{QuickIMP-RED} to reduce it to Problem \ref{problem-op1}.
The next algorithm is a heuristic that attempts to solve this problem. If unsuccesful, Algorithms \ref{FindImp} and \ref{Findredundant} will be applied to further reduce the LP into a smaller one that contains no implied equality and redundant inequality. This will be illustrated in Example \ref{III.5}.
\\
\begin{breakablealgorithm}
	\caption{Heuristic search for a conic combination}
	\label{Heuristic search}
	\begin{algorithmic}[1]
		\REQUIRE \ Problem \ref{problem-op1}. 
		\ENSURE \  SUCCESSFUL, or UNSUCCESSFUL and a reduced LP.
		\\[0.4cm]
		\STATE Let $J:=E_{1}$,
		$J_2:=\emptyset$.
		\STATE Let $\V(F_1)=[a_{i_1},\ldots,a_{i_{n_3}}]$ and $\C(F_1)=[p_1,\ldots,p_{n_3}]$, where $1\le n_3\le m$ and the coefficient $p_j$ corresponds to the variable $a_{i_j}$ for all~$j\in\mathcal{N}_{n_3}$.
	\WHILE{$({\rm there\ exists}\ p_j<0\ {\rm for\ some}\ j\in\mathcal{N}_{n_3})$ $\wedge$ $(|J|>0)$ $\wedge$ $(a_{i_j}\in V(f)\ {\rm for\ some}\ f\in J)$}
		\STATE Solve $a_{i_j}$ from $f=0$ to yield $a_{i_j}=A_{i_j}$ such that $A_{i_j}$ is a polynomial in $V(f)\backslash \{a_{i_j}\}$.
		\STATE $F_1:=F_1-p_j(a_{i_j}-A_{i_j})$.
		\STATE $J:=subs(a_{i_j}=A_{i_j},J)\backslash\{0\}$.
		\STATE $J_2:=subs(a_{i_j}=A_{i_j},J_2)\cup\{a_{i_j}-A_{i_j}\}$.
		\STATE Update $\V(F_1)$ and $\C(F_1)$.
		\ENDWHILE
	   \IF{there does not exist a negative element in $\C(F_1)$}
	   \STATE // $F_1\ge0$ is obviously implied by $\mathcal{R}(S_1)$.
		\STATE Return `SUCCESSFUL'.
		\ELSE 
\STATE // Need to solve 
\begin{problem}\label{problem-op2}
Prove $F_1\ge0$ subject to $\widetilde{R}(J \cup J_2)$ and $\mathcal{R}(S_1)$.
\end{problem}
	     \STATE // Instead of reducing $F_1$ by $J\cup J_2$ directly, since $J_2$ is already in row echelon form after the WHILE loop, we can simplify the computation as follows.
		\STATE Reduce $F_1$ and $J_2$ by $J$ to obtain the remainder $F_2$ and the remainder set $\widetilde{J}_2$, respectively, and also the RREF of $J$ denoted by $\widetilde{J}$.
		\STATE Let $\widetilde{\mathcal{E}}_1=\widetilde{J}\cup \widetilde{J}_2$, which is an RREF of $\widetilde{\mathcal{R}}^{-1}(E_a)$.
		\STATE // Problem \ref{problem-op2} is reduced to 
		\begin{problem}\label{problem-op3}
			Prove $F_2\ge0$ subject to $\widetilde{R}(\widetilde{\mathcal{E}}_1)$ and $\mathcal{R}(S_1)$.
		\end{problem}
		
		\STATE Apply Algorithms \ref{FindImp} and \ref{Findredundant} to Problem \ref{problem-op3} to obtain a reduction of Problem \ref{problem-op1}:
			\begin{problem}\label{problem-op4}
			Prove $F_3\ge0$ subject to $\widetilde{\mathcal{R}}(\widetilde{\mathcal{E}}_2)$ and $\mathcal{R}(V(\{F_3\} \cup \widetilde{\mathcal{E}}_2))$.
		\end{problem}
	\STATE // Problem \ref{problem-op4} contains no implied equalities and redundant inequalities. Thus we only need to consider the inequality constraints $\mathcal{R}(V(\{F_3\} \cup \widetilde{\mathcal{E}}_2))$ instead of $\mathcal{R}(S_1)$, where $|V(\{F_3\} \cup \widetilde{\mathcal{E}}_2)|\le|S_1|$.
		
		\STATE Return `UNSUCCESSFUL' and Problem \ref{problem-op4}.
   \ENDIF
\end{algorithmic}
\end{breakablealgorithm}
\ 
\\
\ 

Next, we give an example to show that Algorithm \ref{Heuristic search} is not always successful even though the problem is solvable.
In general, different decisions made in the algorithm 
can lead to different outcomes. 

\begin{example}\label{III.5}
	We want to determine whether $F=-\frac{1}{2}a_1-a_2+a_3+a_4+a_5-a_6+a_7+a_9\ge0$ subject to $S_a=\{a_i\ge0, i\in\mathcal{N}_{12}\}$ and $E_a=\{a_1+a_2-a_3-a_4=0,a_1+a_2-a_4+a_9+a_{10}-a_{11}-a_{12}=0, a_6-a_9-a_{10}+a_{11}+a_{12}=0, a_5-2a_6=0,a_7+a_8=0\}$.
	Following Algorithm \ref{Heuristic search}, we give the steps in detail. 
	
	\textbf{Step 1}. Run Algorithm \ref{QuickIMP-RED} to obtain
	
	\textbf{Problem \ref{problem-op1}$(*)$}. Prove $F_1\ge0$ subject to $\widetilde{\mathcal{R}}(E_1)$ and $\mathcal{R}(S_1)$, where $F_1=-\frac{1}{2}a_1-a_2+a_3+a_4+a_6+a_9$,  $E_1=\{a_1+a_2-a_3-a_4,\ a_1+a_2-a_4+a_9+a_{10}-a_{11}-a_{12},\ a_6-a_9-a_{10}+a_{11}+a_{12}\}$, and $S_1=\{a_1,a_2,a_3,a_4,a_6,a_9,a_{10},a_{11},a_{12}\}$.

Here, we use Problem \ref{problem-op1}$(*)$ to denote a special instance of Problem \ref{problem-op1}. Similar notations will apply.
	
	Let $J:=E_1=\{a_1+a_2-a_3-a_4,\ a_1+a_2-a_4+a_9+a_{10}-a_{11}-a_{12},\ a_6-a_9-a_{10}+a_{11}+a_{12}\}$, $J_2:=\emptyset$.
	
 Referring to Line 4 of Algorithm \ref{Heuristic search}, we discuss two possible cases.
	
 Case 1: Assume that we solve $a_2$ from $a_1+a_2-a_3-a_4=0$.
	
	{\bf Step 2.} 
	Solve $a_2$ from $a_1+a_2-a_3-a_4=0$ to obtain $a_2=-a_1+a_3+a_4$.
	
	$F_1:=subs(a_2=-a_1+a_3+a_4, F_1)=\frac{1}{2}a_1+a_6+a_9$.
	
	Then $F\ge0$ is proved.
	
 Case 2: Assume that we solve $a_1$ from $a_1+a_2-a_3-a_4=0$.
	
	{\bf Step 2.} 
	Solve $a_1$ from $a_1+a_2-a_3-a_4=0$ to obtain $a_1=-a_2+a_3+a_4$.
	
	$F_1:=subs(a_1=-a_2+a_3+a_4, F_1)=-\frac{1}{2}(a_2-a_3-a_4)+a_6+a_9$.
	
	$J:=subs(a_1=-a_2+a_3+a_4, J)\backslash\{0\}=\{a_3 + a_9 + a_{10} - a_{11}-a_{12}, a_{6} - a_{9} - a_{10} + a_{11}+a_{12}\}$.
	
	$J_2:=subs(a_1=-a_2+a_3+a_4, J_2)\cup\{a_1+a_2-a_3-a_4\}=\{a_1+a_2-a_3-a_4\}$.
	
 After executing this step, we observe that $a_2\notin V(J)$ and the {\bf while} loop in Algorithm \ref{Heuristic search} is terminated. However, we have not yet solved the problem. Thus, we need to continue with the remaining steps in Algorithm \ref{Heuristic search}.
	
	{\bf Step 3.} Reduce $F_1$ and $J_2$ by $J$ to obtain the remainder $F_2=-\frac{1}{2}(a_2-a_4-3a_9-a_{10}+a_{11}+a_{12})$ and the remainder set $\widetilde{J}_2=\{a_1 + a_2 - a_4 + a_9 + a_{10}- a_{11}-a_{12}\}$, respectively, and also the RREF of $J$ denoted by $\widetilde{J}=\{a_3 + a_9 + a_{10} - a_{11}-a_{12}, a_{6} - a_{9}- a_{10} + a_{11}+a_{12}\}$. 
	
	{\bf Step 4.} Let $\widetilde{\mathcal{E}}_1=\widetilde{J}\cup \widetilde{J}_2=\{a_1 + a_2 - a_4 + a_9 + a_{10}- a_{11}-a_{12}, a_3 + a_9 + a_{10} - a_{11}-a_{12}, a_{6} - a_{9}- a_{10} + a_{11}+a_{12}\}$.
	
	Now the problem becomes
	
	\textbf{Problem \ref{problem-op3}$(*)$}. Prove $F_2\ge0$ subject to $\widetilde{\mathcal{R}}(\widetilde{\mathcal{E}}_1)$ and $\mathcal{R}(S_1)$, where  
	$F_2=-\frac{1}{2}(a_2-a_4-3a_9-a_{10}+a_{11}+a_{12})$, \\
	$\widetilde{\mathcal{E}}_1=\{a_1 + a_2 - a_4 + a_9 + a_{10}- a_{11}-a_{12}, a_3 + a_9 + a_{10} - a_{11}-a_{12}, a_{6} - a_{9}- a_{10} + a_{11}+a_{12}\}$, 
	and $S_1=\{a_1,a_2,a_3,a_4,a_6,a_9,a_{10},a_{11},a_{12}\}$.
	
	{\bf Step 5.}  
	Run Algorithms \ref{FindImp} and \ref{Findredundant} to reduce the problem to
	
	\textbf{Problem \ref{problem-op4}$(*)$}.  Prove $F_3\ge0$ subject to $\widetilde{\mathcal{R}}(\widetilde{\mathcal{E}}_2)$ and $\mathcal{R}(V(\{F_3\} \cup \widetilde{\mathcal{E}}_2))$, where
	 $F_3=\frac{1}{2}a_{1} - a_{10} + a_{11} + a_{12}$
	and $ \widetilde{\mathcal{E}}_2=\{a_9 + a_{10} - a_{11}-a_{12}\}$.
	
	 In this step, all the implied equalities and redundant inequalities in Problem \ref{problem-op3}$(*)$ are removed. The detailed steps of this reduction from Problem \ref{problem-op3}$(*)$ to Problem \ref{problem-op4}$(*)$ are given in Appendix A.
	
	Since $F_3+\widetilde{\mathcal{E}}_2[1]=\frac{1}{2}a_{1} +a_9\ge0$, the above LP is solvable. Thus, $F\ge0$ is provable. 
\end{example}

In Line 4 of Algorithm \ref{Heuristic search}, we need to solve $a_{i_j}$ from $f=0$ for some $f\in J$. Different choices of $a_{i_j}$'s can lead to different outcomes, which has been shown in Case 1 and Case 2 in Example \ref{III.5}. Similarly, different choices of $f\in J$ can also lead to different outcomes. For example, following from Example \ref{III.5}, instead of solving $a_1$ or $a_2$ from $a_1+a_2-a_3-a_4=0$ in Step 2, one can also solve $a_1$ or $a_2$ from $a_1+a_2-a_4+a_9+a_{10}-a_{11}-a_{12}=0$. The details are omitted here.


%
Assume that Algorithm \ref{Heuristic search} outputs `UNSUCCESSFUL' and Problem \ref{problem-op4}, which is a reduction of Problem \ref{problem-op1}. We now present the following algorithm for solving this problem.
\\
\ 
\begin{breakablealgorithm}
	\caption{Solving Problem \ref{problem-op4}}
	\label{FindComb}
	\begin{algorithmic}[1]
		\REQUIRE \ Problem \ref{problem-op4}. 
		\ENSURE \  The statement “Problem \ref{problem-op4} is solvable” is TRUE or FALSE.
		\\[0.4cm]
\STATE Assume that $\widetilde{\mathcal{E}}_2$ has the form $\widetilde{\mathcal{E}}_2=\{a_{k_l}-A_{k_l}, l\in\mathcal{N}_{r}\}$, where $r$ is the rank of $\widetilde{\mathcal{E}}_2$, and $A_{k_l}$'s are linear combinations of the free variables $a_{k_{r+1}},\ldots,a_{k_{t}}$, where $t=|V(\widetilde{\mathcal{E}}_2)|\le m$.
	
	\STATE Let $F_4\equiv F_3+\sum\limits_{l=1}^{r}p_l(a_{k_l}-A_{k_l})$, where $p_l, 1\le l\le r$ are to be determined. Since $F_3$ and $A_{k_l}$'s are in terms of the free variables, we can rewrite $F_4$ as
	$F_4\equiv \sum\limits_{l=1}^{r}p_la_{k_l}+\sum\limits_{l=r+1}^{t}P_l a_{k_{l}}$, where $P_l$'s are linear combinations of  $p_l$'s.
	\STATE	// 	By Theorem \ref{Combofalpha}, Problem \ref{problem-op4} can be proved if and only if $F_4$ can be expressed as a conic combination of $a_i$'s.
	\STATE  Solve the following LP:
	\begin{problem}\label{problem-op5}
	{\rm min}$(0)$\ \ {\rm such\ that}\ \  $p_l\ge0,\ l\in\mathcal{N}_r$ {\rm and} $P_l\ge0,\ l\in\mathcal{N}_{t}\backslash\mathcal{N}_{r}$.
	\end{problem}
	
\IF{Problem \ref{problem-op5} can be solved}
\STATE Declare that “Problem \ref{problem-op4} can be solved” is `TRUE'.
\ELSE
\STATE Declare that “Problem \ref{problem-op4} can be solved” is `FALSE'.
\ENDIF
\RETURN The argument “the Problem \ref{problem-op4} can be solved” is TRUE or FALSE.
\end{algorithmic}
\end{breakablealgorithm}

\  
\\[0.1cm]

\begin{remark}
From Theorem \ref{Combofalpha}, Problem \ref{problem-op} is solvable if and only if $F$ can be transformed into a conic combination of $a_i,\ i\in\mathcal{N}_m$ by $E_{a}$. 
We first apply Algorithm \ref{QuickIMP-RED} to Problem \ref{problem-op} to obtain Problem \ref{problem-op1}, which contains no trivially implied equalities and trivially redundant inequalities.
Then we use Algorithm \ref{Heuristic search}, a heuristic method, to try to find a conic combination for $F_1$. Specifically, we identify a monomial term of $F_1$ with a negative coefficient and use $E_1$ to eliminate the corresponding variable in $F_1$. 
Then we repeat this operation until it can not be done. There can be two possibilities. If a conic combination for $F_1$ is obtained, then the problem is solved. If a conic combination for $F_1$ is not obtained, then the problem may or may not be solvable. Algorithm \ref{FindComb} deals with this case. Here, even if Algorithm \ref{Heuristic search} can not solve the problem directly, it will reduce Problem \ref{problem-op} to an equivalent LP but smaller in size, which can be solved effectively by Algorithm \ref{FindComb}.
\end{remark}

\section{Procedures for proving information inequality and identity}

In this section, we present two procedures for proving information inequalities and identities under the constraint of an inequality set and/or equality set.
They are designed in the spirit of Theorem~\ref{LP-S}.
To simplify the discussion, $H(X_1,X_2,\ldots,X_n)$ will be denoted by $h_{1,2,\ldots,n}$, so on and so forth. 
For a joint entropy $t=h_{i_1,i_2,...,i_n}$, the set $L(t)=\{i_1,i_2,...,i_n\}$ is called the {\it subscript set} of $t$. The following defines an order among the joint entropies.

\begin{define}\label{reverse-lexi-order}
	Let $t_1=h_{i_1,i_2,\ldots,i_{n_1}}$ and $t_2=h_{j_1,j_2,\ldots,j_{n_2}}$ be two joint entropies.
	We write $t_1\succ t_2$ if one of the following conditions is satisfied:
	\begin{enumerate}
		\item[1)] $|L(t_1)|>|L(t_2)|$, 
		\item[2)] $|L(t_1)|=|L(t_2)|$, $i_l=j_l$ for $l=1,\ldots,k-1$ and $i_k>j_k$.
	\end{enumerate}
\end{define}


\subsection{Procedure I: Proving Information Inequalities}

\noindent
\textbf{Input:}\\ 
Objective information inequality: $\bar{F}\ge0$.\\
Elemental information inequalities: $\bar{C}_i\ge0,\ i=1,\ldots,m_1$.\\
Additional constraints: $\bar{C}_j\ge0,\ j=m_1+1,\ldots,m_2$; $\bar{C}_k=0,\ k=m_2+1,\ldots,m_3$.\\
// \ Here, $\bar{F},\ \bar{C}_i,\ \bar{C}_j,$ and $\bar{C}_k$ are linear combination of Shannon's information measures.\\ 

\noindent
\textbf{Output:} A proof of $\bar{F}\ge0$ if it is implied by the elemental inequalities and the additional constriants.

\textbf{Step 1}. Transform $\bar{F}$ and $\bar{C}_i$, $i\in \mathcal{N}_{m_3}$ to homogeneous linear polynomials $\widetilde{F}$ and $\widetilde{C}_i$, $i\in \mathcal{N}_{m_3}$ in joint entropies.

\textbf{Step 2}. Fix the joint entropies' order $h_{1,2,\ldots,n}\succ\cdots\succ h_1$. 
Apply Algorithm \ref{division-algorithm} to reduce the inequality set $\{\widetilde{C}_i\ge0, i\in\mathcal{N}_{m_2}\}$ by the equality set $\{\widetilde{C}_i=0, i\in\mathcal{N}_{m_3}\backslash\mathcal{N}_{m_2}\}$ to obtain the reduced inequality set $\{C_i\ge0,i\in\mathcal{N}_{m}\}$.

\textbf{Step 3}. Reduce $\widetilde{F}$ by the equality set $\{\widetilde{C}_i=0, i\in\mathcal{N}_{m_3}\backslash\mathcal{N}_{m_2}\}$ to obtain $F_5$.

// We need to solve
\begin{problem}\label{problem-p8}
 Prove $F_5\ge0$ under the constraints $C_i\ge0$, $i\in\mathcal{N}_{m}$.
\end{problem}

\textbf{Step 4}. Under the variable order $h_{1,2,\ldots,n}\succ\cdots\succ h_1\succ a_1\succ\cdots \succ a_m$, apply Algorithm \ref{Algorithm-transf.} to Problem \ref{problem-p8} to obtain

\textbf{Problem \ref{problem-p5}$(*)$}. Prove $F\ge0$ subject to $\widetilde{\mathcal{R}}(J_1)$ and $a_i\ge0,\ i\in\mathcal{N}_{m}$, where $J_1=\{{f}_i,\ i\in\mathcal{N}_{m_4}\}$. 

\textbf{Step 5}. Apply Algorithm \ref{QuickIMP-RED} and Algorithm \ref{Heuristic search} successively to the above problem.
If Algorithm \ref{Heuristic search} outputs `SUCCESSFUL', then the objective function $\bar{F}\ge0$ is proved. Otherwise, the following reduced LP is obtained:

\textbf{Problem \ref{problem-op4}$(*)$}. Prove $F_3\ge0$ subject to $\widetilde{\mathcal{R}}(\widetilde{\mathcal{E}}_2)$ and $\mathcal{R}(V(\{F_3\} \cup \widetilde{\mathcal{E}}_2))$, 
	where $\widetilde{\mathcal{E}}_2=\{\tilde{f}_i,\ i\in\mathcal{N}_{m_5}\}$. 

// Note that $m_5\le m_4$ and $|\mathcal{R}(V(\{F_3\} \cup \widetilde{\mathcal{E}}_2))|\le m$.

\textbf{Step 6}. Apply Algorithm \ref{FindComb} to the above problem. If Algorithm \ref{FindComb} outputs `TRUE', then the objective function $\bar{F}\ge0$ is proved. 
Otherwise, declare `Not Provable'.

\subsection{Procedure II: Proving Information Identities}
		
		\noindent
		\textbf{Input:} \\
		Objective information identity: $\bar{F}=0$.\\
   Elemental information inequalities: $\bar{C}_i\ge0,\ i=1,\ldots,m_1$.\\
		Additional constraints: $\bar{C}_j\ge0,\ j=m_1+1,\ldots,m_2$; $\bar{C}_k=0,\ k=m_2+1,\ldots,m_3$; \\
		//\ \ Here, $\bar{F},\ \bar{C}_i,\ \bar{C}_j,$ and $\bar{C}_k$ are linear combinations of information measures.
		
		\noindent
		\textbf{Output:}
		A proof of $\bar{F}=0$ if it is implied by the elemental inequalities and the additional constraints.
       
      \textbf{ Step 1}. Transform $\bar{F}$ and $\bar{C}_i$, $i\in \mathcal{N}_{m_3}$ to homogeneous linear polynomials $\widetilde{F}$ and $\widetilde{C}_i$, $i\in \mathcal{N}_{m_3}$ in joint entropies.
       
     \textbf{  Step 2}. Fix the joint entropies' order $h_{1,2,\ldots,n}\succ\cdots\succ h_1$. 
       Apply Algorithm \ref{division-algorithm} to reduce the inequality set $\{\widetilde{C}_i\ge0, i\in\mathcal{N}_{m_2}\}$ by the equality set $\{\widetilde{C}_i=0, i\in\mathcal{N}_{m_3}\backslash\mathcal{N}_{m_2}\}$ to obtain the reduced inequality set $\{C_i\ge0,i\in\mathcal{N}_{m}\}$.
       
     \textbf{  Step 3}. Reduce $\widetilde{F}$ by the equality set $\{\widetilde{C}_i=0, i\in\mathcal{N}_{m_3}\backslash\mathcal{N}_{m_2}\}$ to obtain $F_6$.
       
       // We need to solve
       \begin{problem}\label{problem-p12}
       	Prove $F_6\ge0$ under the constraints $C_i\ge0$, $i\in\mathcal{N}_{m}$.
       \end{problem}
   
   \textbf{ Step 4}. Under the variable order $h_{1,2,\ldots,n}\succ\cdots\succ h_1\succ a_1\succ\cdots \succ a_m$, apply Algorithm \ref{Algorithm-transf.} to Problem \ref{problem-p12} to obtain

\textbf{Problem \ref{problem-p5}$(*)$}. Prove $F\ge0$ subject to $\widetilde{\mathcal{R}}(J_1)$ and $a_i\ge0,\ i\in\mathcal{N}_{m}$, where $J_1=\{{f}_i,\ i\in\mathcal{N}_{m_4}\}$. 

\textbf{Step 5}. Apply Algorithm \ref{FindImp} to the above problem to obtain a reduced and pure LP:

\textbf{Problem \ref{pureproblem}$(*)$}. Prove $F_1$ subject to $\widetilde{R}(E_1)$ and $\mathcal{R}(V(\{F_1\}\cup E_1))$.

If $F_{1}\equiv 0$, then the objective function $\bar{F}=0$ is proved.  Otherwise, declare `Not Provable'.

Next, we give an example to show the effectiveness of our procedure.
\begin{example}\label{example-1}
$I(X_i;X_4)=0,\ i=1,2,3$ and $H(X_4|X_i,X_j)=0, 1\leq i<j\leq 3$ $\Rightarrow$ $H(X_i)\geq H(X_4)$.
\end{example}

This example has been discussed in \cite{Guo-Yeung-Gao2023}. Due to the symmetry of this problem, we only need to prove $H(X_1) \ge H(X_4)$.
Next we give the proof based on Procedure I.

\textbf{Step 1.} We need to solve the problem:\\
Prove
$$\widetilde{F}=h_1-h_4$$
under constraints $\widetilde{C}_i\ge0,\ i=1,\ldots,28$, $\widetilde{C}_i=0,\ i=29,\ldots,34$. 

\textbf{Steps 2-3.} After reduction, the above problem becomes:\\
Prove
$$
F=h_1-h_4
$$

under the constraints $C_{i}\ge0,\ i=1,\ldots,18$, where
\begin{equation}
	\begin{array}{ll}
C_1=h_{4},
C_2=h_{1} + h_{2} - h_{1, 2},
C_3=h_{1} + h_{2} + h_{4} - h_{1, 2},
C_4=h_{1} + h_{3} - h_{1, 3},
C_5=h_{1} + h_{3} + h_{4} - h_{1, 3},\\
C_6=h_{2} + h_{3} - h_{2, 3},
C_7=h_{2} + h_{3} + h_{4} - h_{2, 3},
C_8=-h_{1} + h_{1, 2} + h_{1, 3} -   h_{1, 2, 3},
C_9=-h_{2} + h_{1, 2} + h_{2, 3} - h_{1, 2, 3},\\
C_{10}=-h_{3} + h_{1, 3} +   h_{2, 3} - h_{1, 2, 3},
C_{11}=-h_{1} - h_{4} + h_{1, 2} + h_{1, 3} -   h_{1, 2, 3, 4},
C_{12}=-h_{2} - h_{4} + h_{1, 2} + h_{2, 3} -   h_{1, 2, 3, 4},\\
C_{13}=-h_{3} - h_{4} + h_{1, 3} + h_{2, 3} - h_{1, 2, 3, 4},
C_{14}=h_{1, 2, 3} - h_{1, 2, 3, 4},
C_{15}=-h_{1, 2} + h_{1, 2, 3, 4},
C_{16}=-h_{1, 3} +   h_{1, 2, 3, 4},\\
C_{17}=-h_{2, 3} + h_{1, 2, 3, 4},
C_{18}=-h_{1, 2, 3} + h_{1, 2, 3, 4}.
\end{array}\end{equation}

\textbf{Step 4}. Apply Algorithm \ref{Algorithm-transf.} to the above problem to obatin a reduced problem:

\textbf{Problem \ref{problem-p5}$(*)$}. Prove $F\ge0$ subject to $\widetilde{\mathcal{R}}(J_1)$ and $a_i\ge0,\ i\in\mathcal{N}_{18}$, where $F=a_{6}+a_{7} + a_{14}$ and $J_1=\{{f}_i,\ i\in\mathcal{N}_{9}\}$. 


\textbf{Step 5}. Since $F$ is a conic combination of $a_i,\ i\in\mathcal{N}_{18}$, $\bar{F}\ge0$ is provable.
\\[0.3cm]

The aforementioned problem is proved without solving any LPs. If we want to further find a reduced minimal problem, then we can apply Algorithm \ref{FindImp} and \ref{Findredundant} to Problem \ref{problem-p5} to obtain the following LP that contains no implied equality and redundant inequality:

Prove $F_1$ subject to $\widetilde{R}(E_1)$ and $\mathcal{R}(V(\{F_1\}\cup E_1))$,
where $E_1=\{\tilde{f}_1,\tilde{f}_2\}$,
\begin{equation}\begin{array}{ll}
F_1=a_{6} - a_{11} + a_{12} + a_{13} + a_{17},\\
\tilde{f}_1=a_4- a_6+ a_{11} - a_{12},\\
\tilde{f}_2= a_2 - a_6 + a_{11} - a_{13}.
\end{array}\end{equation}
Since $F_1+\tilde{f}_2= a_2 + a_{12} + a_{17}$, we show that $F_1\ge0$. Thus an explicit proof is given.

%
%
%
%
%
%
%
%
%

\section{Two Applications}
\label{sec:appl}
In this section, we will present two applications of our method. The first one is an information inequality proved in \cite{Dougherty-Freiling-Zeger2007}. The second one is  the example used in \cite{Guo-Yeung-Gao2023}, in which we can significantly reduce the required computation for solving the LP. 

\subsection{Dougherty-Freiling-Zeger's Problem}

The information theoretic inequality needs to be proved in \cite{Dougherty-Freiling-Zeger2007} is specified by an LP with $8$ random variables.

\textbf{Problem} $\mathbf{P_1}$:
Prove $I(B;D,X,Z) \le I(W;A,B,C,D)$ under the constraints
\begin{equation}\label{DFZ-cons}
\begin{array}{ll}
I(W;A,B,C,D) = I(X;A,B,W),\ 
I(W;A,B,C,D) = I(Y;B,C,X),\ 
I(W;A,B,C,D) = I(Z;C,D,Y),\\ 
I(A;B,C,D,Z) = I(B;D,X,Z),\ 
I(B;A,D,W,Z) = I(B;D,X,Z),\ 
I(C;A,D,W,Z) = I(B;D,X,Z),\\
I(D;A,B,C,Y) = I(B;D,X,Z),\ 
I(C;A,W,Y) = I(B;D,X,Z),\ 
I(B;A) = 0,\ 
I(C;A,B) = 0,\\
I(D;A,B,C) = 0,\ 
I(X;C,D|A,B,W) = 0,\ 
I(Y;A,D,W|B,C,X) = 0,\ 
I(Z;A,B,W,X|C,D,Y) = 0.
\end{array}
\end{equation}

We now solve \textbf{Problem}~$\mathbf{P_1}$ by using Procedure I. 

\noindent
\textbf{Input:} \\
Objective information inequality: $\bar{F}=I(W;A,B,C,D)-I(B;D,X,Z)\geq 0$.\\
Inequality Constraints: 
the elemental information inequalities generated by random variables $A,B,C,D,X,Y,Z,W$ (totally $1800$ inequalities).\\
Equality Constraints:  totally $14$ equalities in \eqref{DFZ-cons}.

\textbf{Step 1}. The variable vector generated from $A,B,C,D,X,Y,Z,W$ has $2^{8}-1$ elements (joint entropies). 
Transform $\bar{F}$ into the joint entropy form $\widetilde{F}$.
Express the elemental information inequalities in terms of the joint entropies to obtain $\widetilde{C}_{i},\ i\in\mathcal{N}_{1800}$. Likewise, express the equality constraints in \eqref{DFZ-cons} in terms of the joint entropies to obtain $\widetilde{C}_{i}=0,\ i\in\mathcal{N}_{1814}\backslash\mathcal{N}_{1800}$.

\textbf{Step 2}. Apply Algorithm 1 to reduce $\{\widetilde{C}_{i},\ i\in\mathcal{N}_{1800}\}$ by $\{\widetilde{C}_{i},i\in\mathcal{N}_{1814}\backslash\mathcal{N}_{1800}\}$  to obtain $\{C_i\ge0,\ i\in\mathcal{N}_{1793}\}$.


\textbf{Step 3}. Reduce $\widetilde{F}$ by $\{\widetilde{C}_{i},i\in\mathcal{N}_{1814}\backslash\mathcal{N}_{1800}\}$ to obtain

 $F_5=h_{1, 2, 3} - h_{1, 2, 3, 4} - h_{1, 2, 3, 7} + h_{1, 2, 3, 4, 7} + h_{1, 2, 3, 4, 5, 6, 7} + h_{8} - h_{1, 2, 3, 4, 5, 6, 7, 8}$.

We need to solve 

\textbf{Problem \ref{problem-p8}$(*)$}. Prove $F_5\ge0$ under the constraints $C_i\ge0,i\in\mathcal{N}_{1793}$.

\textbf{Step 4}. 
Apply Algorithm 2 to obtain a reduced LP:

\textbf{Problem \ref{problem-p5}$(*)$}. Prove $F\ge0$ subject to $\widetilde{\mathcal{R}}(J_1)$ and $a_i\ge0,\ i\in\mathcal{N}_{1793}$, where $J_1=\{{f}_i,\ i\in\mathcal{N}_{1559}\}$. 

\textbf{Step 5}. Apply Algorithm \ref{QuickIMP-RED} and Algorithm  \ref{Heuristic search} to the above problem successively.
Algorithm \ref{Heuristic search} outputs `SUCCESSFUL'. Thus $\bar{F}\ge0$ is provable.

In other words, we show that $F$ can be transformed to a conic combination of $\{a_{i}\ge0,\ i\in\mathcal{N}_{1793}\}$ by $\{f_i=0, i\in\mathcal{N}_{1559}\}$.
This is given by
$$\begin{array}{ll}
	F_3=\frac{1}{2}(a_{24} + a_{28} + a_{35} + a_{129} + a_{185} + a_{520} + a_{1048} + a_{1053} + a_{1187} + a_{1237} + a_{1556} + a_{1628} + a_{1681} + a_{1782}).
\end{array}$$
Thus $\bar{F}\ge0$ is provable.

Table I shows the advantage of Procedure I for this example by comparing it with the Direct LP method induced by Theorem II.2, with ITIP, and with the procedure in \cite{Guo-Yeung-Gao2023}. 
Note that the procedure in \cite{Guo-Yeung-Gao2023} can reduce this example to the minimum LP to the greatest extent possible. However, we even do not need to solve LP by Procedure I in this work.

\begin{table*}[h!]
\caption{}
\label{table-compare0}
\begin{center}
\begin{tabular}{ |c|c|c|c|} 
 \hline
  & Number of variables & Number of equality constraints & Number of Inequality constraints \\ 
 \hline
Direct LP method & 255 & 14 & 1800 \\ 
 \hline
ITIP & 241  & 0 & 1800 \\ 
 \hline
LP obtained in \cite{Guo-Yeung-Gao2023} & 206& 0 & 1673 \\ 
\hline
This work &   \multicolumn{3}{c|}{no LP needs to be solved}  \\ 
\hline
\end{tabular}
\end{center}
\end{table*}

\subsection{Tian's Problem}

The framework of regenerating codes introduced in the seminal work of Dimakis {\it et al.} \cite{Dimakis2010}
addresses the fundamental tradeoff between the storage and repair bandwidth 
in erasure-coded distributed storage systems.
In \cite{Tian2014}, a new outer bound on the rate region for $(4,3,3)$ exact-repair
regenerating codes was obtained. This outer bound was proved by means of a 
computational approach built upon the LP framework in \cite{Yeung1997} for proving Shannon-type 
inequalities. The LP that needs to be solved, however, is exceedingly large. In order to make the computation manageable, Tian took advantage of the symmetry of the problem and other problem-specific structures to reduce the numbers of variables and constraints in the LP. This outer bound not only can provide a complete characterization of 
the rate region, but also proves the existence of a non-vanishing gap between the optimal tradeoff of the exact-repair codes and that of the functional-repair codes for the 
parameter set $(4,3,3)$.
It was the first time that a non-trivial information theory problem was solved using this LP framework.\footnote{It was subsequently proved analytically by Sasidharan {\it et al.} \cite{Sasidharan2016} that the same holds for every parameter set.}

In this work, we apply the results in the previous sections to Tian's problem, and give a simpler proof by our new method.
We first give the abstract formulation of the problem.


%


\begin{define}\label{permutation}
A permutation $\pi$ on the set $\mathcal{N}_4$ is a one-to-one mapping $\pi$: $\mathcal{N}_4\rightarrow\mathcal{N}_4$. The collection of all permutations is denoted as $\prod$.
\end{define}

In the problem formulation, we consider $16$ random variables grouped into the following two sets:

$$\begin{array}{ll}
\mathcal{W}\!\!\!\!&=\{W_1,W_2,W_3,W_4\}, \\[0.1cm]
\mathcal{S}&=\{S_{1,2},S_{1,3},S_{1,4},S_{2,1},S_{2,3},S_{2,4},S_{3,1},S_{3,2},S_{3,4},S_{4,1},S_{4,2},S_{4,3}\}.
\end{array}$$

A permutation $\pi$ on $\mathcal{N}_4$ is applied to map one random variable to another random variable.
For example, the permutation $\pi(1,2,3,4)=(2,3,1,4)$ maps the random variable $W_1$ to $W_2$.
Similarly it maps the random variable $S_{i,j}$ to $S_{\pi(i),\pi(j)}$. When $\pi$ is applied to a set of random variables, the permutation is applied to every random variable in the set. For example for the aforementioned permutation $\pi$, we have $\pi(W_1,\ S_{2,3})=(W_2,\ S_{3,1})$. 

%

The original problem is \\
\textbf{Problem} $\mathbf{P_6}$:  Prove
\begin{equation}\label{OB}
4\alpha+6\beta\geq 3B
\end{equation} 
under the constraints
\begin{itemize}
\item[C1] $H(\pi(\mathcal{A}),\pi(\mathcal{B}))=H(\mathcal{A},\mathcal{B})$, for any sets $\mathcal{A}\subseteq\mathcal{S}$ and $\mathcal{B}\subseteq\mathcal{W}$ and any permutation $\pi\in\prod$,
\item[C2] $H(\mathcal{W}\cup\mathcal{S}|\mathcal{A})=0$, any $\mathcal{A}\subseteq\mathcal{W}:|\mathcal{A}|=3$,
\item[C3] $H(S_{i,j}|W_i)=0,\ j\in \mathcal{N}_4,\ i\in\mathcal{N}_4\backslash\{j\}$,
\item[C4] $H(W_j|\{S_{i,j}\in\mathcal{S}: i\in \mathcal{N}_n\backslash\{j\}\})=0$, for any $j\in\mathcal{N}_4$,
\item[C5] $H(\mathcal{W}\cup\mathcal{S})=B$,
\item[C6]$H(\mathcal{A})=B$, for any $\mathcal{A}$ such that $|\mathcal{A}\cap\mathcal{W}|\geq3$,
\item[C7] $H(W_i)\le \alpha$, $W_i\in\mathcal{W}$,
\item[C8] $H(S_{i,j})\le\beta$, $S_{i,j}\in\mathcal{S}$.
\end{itemize}

For this specific problem, the random variables involved exhibit strong symmetry due to the setup of the problem. To reduce the scale of the problem, Tian proved in \cite[Section III-B]{Tian2014} that only a subset of the random variables in $\mathcal{W}\cup\mathcal{S}$ is needed for solving \textbf{Problem} $\mathbf{P_6}$. A similar idea was also used in \cite{Tian2015,Tian2022}.

According to Tian's proof in Section III-B of \cite{Tian2014}, \textbf{Problem} $\mathbf{P_6}$ can be reduced to the following simpler problem,\\
\textbf{Problem} $\mathbf{P_7}$: Prove 
\begin{equation}\label{OB}
4\alpha+6\beta\geq 3B
\end{equation} 
under the constraints: C1, C3, C4, C6, C7 and C8 on the $12$ random variables in the set
$$\begin{array}{ll}
\mathcal{W}_1\cup\mathcal{S}_1=\{W_1,W_2,W_4\}\cup \{S_{2,1},S_{3,1},S_{4,1},S_{1,2},S_{3,2},S_{4,2},S_{1,4},S_{2,4},S_{3,4}\}.
\end{array}$$



\begin{remark}\label{h-entropy}
In the following computation, in order to simplify the notation, we will use, for example, 
$h_{1,2,3,4,5,6,7,8,9,10,11,12}$ to represent the joint entropy {\small$H(W_1,W_2,W_4, S_{2,1},S_{3,1},S_{4,1},S_{1,2},S_{3,2},S_{4,2},S_{1,4},S_{2,4},S_{3,4})$}. Similarly, we will use $h_{1}$ to represent $H(W_1)$, $h_{2,5}$ to represent $H(W_2,S_{3,1})$, so on and so forth.
\end{remark}

We now solve \textbf{Problem}~$\mathbf{P_7}$ by using Procedure I. 

\noindent
\textbf{Input:} \\
Objective information inequality: $\bar{F}=4\alpha+6\beta-3B\geq 0$.\\
Inequality Constraints: 
the elemental information inequalities generated by random variables $\mathcal{W}_1\cup\mathcal{S}_1$ (total $67596$ inequalities); 
C7 and C8 (total $12$ inequalities).\\
Equality Constraints: C1, C3, C4 and C6 (total $22945$ equalities).\\

\textbf{Step 1}. The variable vector generated from $\mathcal{W}_1\cup\mathcal{S}_1$ has $2^{12}-1$ elements (joint entropies). Express C7, C8 and the elemental information inequalities in terms of the joint entropies to obtain $\widetilde{C}_{i},\ i\in\mathcal{N}_{67608}$.
According to conditions C1, C3, C4 and C6, write equality constraints in joint entropy forms: $\widetilde{C}_{i}=0,\ i\in\mathcal{N}_{90553}\backslash\mathcal{N}_{67608}$.

\textbf{Step 2}. Apply Algorithm 1 to reduce $\{\widetilde{C}_{i},\ i\in\mathcal{N}_{67608}\}$ by $\{\widetilde{C}_{i},i\in\mathcal{N}_{90553}\backslash\mathcal{N}_{67608}\}$  to obtain  $\{C_i\ge0,\ i\in\mathcal{N}_{10189}\}$.


\textbf{Step 3}. Reduce $\widetilde{F}$ by $\{\widetilde{C}_{i},i\in\mathcal{N}_{90553}\backslash\mathcal{N}_{67608}\}$ to obtain $F_5=4\alpha+6\beta-3h_{2,3,5,6,7,8,10}$.

We need to solve 

\textbf{Problem \ref{problem-p8}$(*)$}. Prove $F_5\ge0$ under the constraints $C_i\ge0,i\in\mathcal{N}_{10189}$.

\textbf{Step 4}. 
%
Apply Algorithm \ref{Algorithm-transf.} to obtain a reduced LP:

\textbf{Problem \ref{problem-p5}$(*)$}. Prove $F\ge0$ subject to $\widetilde{\mathcal{R}}(J_1)$ and $a_i\ge0,\ i\in\mathcal{N}_{10189}$, where $J_1=\{{f}_i,\ i\in\mathcal{N}_{9859}\}$. 

\textbf{Step 5}. Apply Algorithm \ref{QuickIMP-RED} and Algorithm  \ref{Heuristic search} to the above problem successively.
Algorithm  \ref{Heuristic search} outputs `SUCCESSFUL'. Thus $\bar{F}\ge0$ is provable.

In other words, we show that $F$ can be transformed to a conic combination of $\{a_{i},\ i\in\mathcal{N}_{10189}\}$ by  $\{f_i=0, i\in\mathcal{N}_{9859}\}$.
This is given by 
\begin{equation}\label{Final-1}
\begin{array}{ll}
F_3=7 a_{6} + 2 a_{85} + 4 a_{94} + a_{119} + a_{167} + 3 a_{169} + 4 a_{211} + a_{223} + a_{290} + a_{335} + a_{340} + a_{353} + 4 a_{450} + 4 a_{484}\\
~~~ + a_{519} + a_{667} + a_{727} + 4 a_{735} + a_{819} + a_{820} + a_{827} + a_{829} + a_{859} + a_{868} + 3 a_{906} + a_{916}+4a_{10188}+6a_{10189}.
\end{array}
\end{equation}
The fomulas used above is listed in Appendix B.

Table \ref{table-compare1} shows the advantage of Procedure~I for Tian's problem by comparing it with the Direct LP method induced by Theorem~\ref{LP-S}, ITIP, Tian's method in \cite{Tian2014}, and our previous work in \cite{Guo-Yeung-Gao2023}.
\begin{table*}[h!]
\caption{}
\label{table-compare1}
\begin{center}
\begin{tabular}{ |c|c|c|c|} 
 \hline
  & Number of variables & Number of equality constraints & Number of Inequality constraints \\ 
 \hline
Direct LP method & 4098 & 22945 & 67608 \\ 
 \hline
ITIP & 600  & 0 & 67608 \\ 
 \hline
Tian's Method & 176 & 0 & 6152 \\ 
 \hline
LP in \cite{Guo-Yeung-Gao2023} & 101& 0 & 649 \\ 
\hline
This work &   \multicolumn{3}{c|}{no LP needs to be solved}  \\ 
\hline
\end{tabular}
\end{center}
\end{table*}

Note that even if Algorithm \ref{Heuristic search} cannot obtain a conic combination as desired, 
it can still reduce the problem to the minimal LP in a shorter time and with less memory compared with our previous work \cite{Guo-Yeung-Gao2023}. Table \ref{table-compare2} shows the advantage of Procedure~I for reducing Tian's problem by comparing it with the procedure in \cite{Guo-Yeung-Gao2023}.
In Table \ref{table-compare2}, ``Time'' and ``Memory'' refer to the time and memory it takes to simplify the original LP to the minimal LP, respectively. The experiment results are obtained by MAPLE running on a desktop PC with an i7-6700 Core, 3.40GHz CPU and 16G memory. 

\begin{table*}[h!]
	\caption{}
	\label{table-compare2}
	\begin{center}
		\begin{tabular}{ |c|c|c|c|c|} 
			\hline
			& Number of variables  & Number of Inequality constraints & Time & Memory \\ 
			\hline
		LP in  \cite{Guo-Yeung-Gao2023} & 101 & 649 &  23000s & 900M\\
		\hline
			LP in this work &  101 & 649 & 33s & 70M \\ 
			\hline
		\end{tabular}
	\end{center}
\end{table*}

\section{Conclusion and discussion}
\label{sec-conc}
As discussed in Section III, since different elimination choices of variables can lead to different results, this heuristic method (Algorithm \ref{Heuristic search}) may not necessarily succeed. Nevertheless, if the first attempt is unsuccessful, we can repeat the attempt with different elimination choices of variables for a certain maximum number of times.
Next, we summarize in Table \ref{table-compare3} some experimental results on the effectiveness of Algorithm 4 for solving various problems. In the table, ``TS'' denotes how many times need to run Algorithm \ref{Heuristic search} to get a successful result, ``TSH'' denotes how many times out of a hundred runs of Algorithm 4 are successful, and ``Time'' denotes the time required to repeatedly run Algorithm \ref{Heuristic search} to obtain a successful result.


\begin{table*}[h!]
	\caption{}
	\label{table-compare3}
	\begin{center}
		\begin{tabular}{ |c|c|c|c|} 
			\hline
			& TS  & TSH &  Time \\ 
			\hline
			Example \ref{III.5} & 2 & 52 &  2s \\
			\hline
			Example \ref{example-1} & 1 & 100 &  0.2s \\
			\hline
			Dougherty-Freiling-Zeger's problem &  3 & 32 & 11s \\ 
			\hline
		 Tian's problem &  12 & 10 & 80s \\ 
			\hline
		\end{tabular}
	\end{center}
\end{table*}

The data given in Table \ref{table-compare3}  is for reference only. The problems listed in the table are all solvable. For problems that are not solvable, we have to use Algorithm \ref{FindComb} to solve an LP which typically has a much smaller size compared with the original problem, and Algorithm \ref{FindComb} will output `FALSE'. Compared with \cite{Guo-Yeung-Gao2023}, the method here for obtaining the reduced minimal characterization set is considerably simpler.

To end this paper, we put forth the following conjecture on the effectiveness of Algorithm \ref{Heuristic search}:

\textbf{Conjecture.} If the problem is solvable, then there exists at least one ordering of the variables such that Algorithm \ref{Heuristic search} outputs `SUCCESSFUL'.

\newpage
\section*{Appendices}

\subsection{Two enhancements of Algorithm \ref{QuickIMP-RED}}

In this section, we present two algorithms as enhancements of Algorithm \ref{QuickIMP-RED}. We call Problem \ref{problem-op} a pure LP  if it contains no implied equality,  and call it a minimal LP if it contains no redundant inequality.  


First, we give a general algorithm for reducing Problem \ref{problem-op} to a pure LP.
\ 
\\
\begin{breakablealgorithm}
	\caption{Pure LP Algorithm}
	\label{FindImp}
	\begin{algorithmic}[1]
		\REQUIRE \ Problem \ref{problem-op}.
		\ENSURE \  A pure LP.
		\\[0.4cm]
		\STATE Use Algorithm \ref{QuickIMP-RED} to reduce Problem \ref{problem-op} to
		\STATE \textbf{Problem \ref{problem-op1}}. Proving $F_1\ge0$ subject to $\widetilde{\mathcal{R}}(E_1)$ and $\mathcal{R}(S_{1})$.
		\FOR{$i$ from $1$ to $m$}
		\STATE Apply Algorithm \ref{Heuristic search} to solve the following LP:
		\begin{problem}\label{problem-p6}
			Prove $-a_i\ge0$ subject to $\widetilde{\mathcal{R}}(E_1)$ and $\mathcal{R}(S_{1})$.
		\end{problem}
		\IF{Algorithm \ref{Heuristic search} outputs `SUCCESSFUL'}
		\STATE $E_1:=subs(a_i=0,E_1)\backslash\{0\}$.
	\STATE $S_{1}:=S_{1}\backslash\{a_i\}$.
		\STATE $F_1:=subs(a_i=0,F_1)$.
		\ELSE
		\STATE Algorithm \ref{Heuristic search} outputs a reduced LP. Apply Algorithm \ref{FindComb} to solve this LP.
		\IF{Algorithm  \ref{FindComb} outputs `TRUE'}
		\STATE $E_1:=subs(a_i=0,E_1)\backslash\{0\}$.
	\STATE $S_{1}:=S_{1}\backslash\{a_i\}$.
		\STATE $F_1:=subs(a_i=0,F_1)$.
		\ENDIF
		\ENDIF
		\ENDFOR
\STATE Reduce $F_1$ by $E_1$ to obtain the remainder $F_2$ and the RREF of $E_1$, $E_2$.
		\RETURN 
		A pure LP:
		\begin{problem}\label{pureproblem}
			Prove $F_2$ subject to $\widetilde{R}(E_2)$ and $\mathcal{R}(V(\{F_2\}\cup E_2))$.
		\end{problem}		
	\end{algorithmic}
\end{breakablealgorithm}
\ 
\\[0.2cm]
Next, we give a general algorithm for finding a minimal LP from Problem  \ref{problem-op}.
\ 
\\
\begin{breakablealgorithm}
	\caption{Minimal LP Algorithm}
	\label{Findredundant}
	\begin{algorithmic}[1]
		\REQUIRE \ Problem \ref{problem-op}.
		\ENSURE \  A minimal LP.
		\\[0.4cm]
		\STATE Run Algorithm \ref{QuickIMP-RED} to reduce Problem \ref{problem-op} to
		\STATE \textbf{Problem \ref{problem-op1}}. Prove $F_1\ge0$ subject to $\widetilde{\mathcal{R}}(E_1)$ and $\mathcal{R}(S_1)$.
		\FOR{$i$ from $1$ to $m$}
		\STATE Let $\bar{S}_{a}=S_1\backslash\{a_i\}$.
		\IF{$a_i\in V(f)$ for some $f \in E_1$}
		\STATE Solve $a_i$ from $f=0$ to obtain $a_i=A_i$. 
		\ELSE
		\STATE Let $A_i$ be $a_i$.
		\ENDIF
		\STATE $\bar{E}_{a}:=subs(a_i=A_i,E_1)\backslash\{0\}$.     
		\STATE Run Algorithm \ref{Heuristic search} to solve the following LP:
		\begin{problem}\label{problem-p7}
			Prove $A_i\ge0$ subject to $\widetilde{\mathcal{R}}(\bar{E}_{a})$ and $\mathcal{R}(\bar{S}_{a})$.
		\end{problem}
	\IF{Algorithm \ref{Heuristic search} outputs `SUCCESSFUL'}
		\STATE $E_1:=subs(a_i=A_i,E_1)\backslash\{0\}$.
	\STATE $S_1:=S_1\backslash\{a_i\}$.
	\STATE $F_1:=subs(a_i=A_i,F_1)$.
	\ELSE 
	\STATE Algorithm \ref{Heuristic search} outputs a reduced LP. Apply Algorithm \ref{FindComb} to this LP.
	\IF{Algorithm  \ref{FindComb} outputs `TRUE'}
		\STATE $E_1:=subs(a_i=A_i,E_1)\backslash\{0\}$.
		\STATE $S_1:=S_1\backslash\{a_i\}$.
		\STATE $F_1:=subs(a_i=A_i,F_1)$.
		\ENDIF
		\ENDIF
		\ENDFOR
\STATE Reduce $F_1$ by $E_1$ to obtain the remainder $F_2$ and the RREF of $E_1$, $E_2$.
		\RETURN 
		A minimal LP:
		\begin{problem}
			Proving $F_2$ subject to  $\widetilde{\mathcal{R}}(E_2)$ and $\mathcal{R}(V(\{F_2\}\cup E_2))$.
		\end{problem}
	\end{algorithmic}
\end{breakablealgorithm}
\ 
\\
%

Next, we give the detailed steps of the reduction from Problem \ref{problem-op3}$(*)$ to Problem \ref{problem-op4}$(*)$. 



// We first follow Algorithm \ref{FindImp}.

\textbf{Step 1}. Use Algorithm \ref{QuickIMP-RED} to reduce Problem \ref{problem-op3}$(*)$ to

 \textbf{Problem \ref{problem-op1}$(*)$}. Proving $F_1\ge0$ subject to $\widetilde{\mathcal{R}}(E_1)$ and $\mathcal{R}(S_{1})$, 
 where $F_1=-\frac{1}{2}(a_2-a_4-3a_9-a_{10}+a_{11}+a_{12})$, 
 $E_1=\{a_1 + a_2 - a_4 + a_9 + a_{10}- a_{11}-a_{12}, a_3 + a_9 + a_{10} - a_{11}-a_{12}, a_{6} - a_{9}- a_{10} + a_{11}+a_{12}\}$, 
 and\\ $S_{1}=\{a_1,a_2,a_3,a_4,a_6,a_9,a_{10},a_{11},a_{12}\}$.

\textbf{Step 2}. For $i\in\mathcal{N}_{12}\backslash\{5,7,8\}$, run Algorithm \ref{Heuristic search} to solve the following LP:

\textbf{Problem \ref{problem-p6}$(*)$}. Prove $-a_i\ge0$ subject to $\widetilde{\mathcal{R}}(E_1)$ and $\mathcal{R}(S_{1})$.

\textbf{Step 3}. Algorithm \ref{Heuristic search} outputs `SUCCESSFUL' when $i=3$, then let

 $E_1=subs(a_3=0,E_1)\backslash\{0\}=\{a_1 + a_2 - a_4 + a_9 + a_{10}- a_{11}-a_{12}, a_9 + a_{10} - a_{11}-a_{12}, a_{6} - a_{9}- a_{10} + a_{11}+a_{12}\}$, 
 
 $S_{1}=S_{1}\backslash\{a_3\}=\{a_1,a_2,a_4,a_6,a_9,a_{10},a_{11},a_{12}\}$, 
 
 $F_1=subs(a_i=0,F_1)=-\frac{1}{2}(a_2-a_4-3a_9-a_{10}+a_{11}+a_{12})$.
 
 \textbf{Step 4}. Algorithm \ref{Heuristic search} outputs `SUCCESSFUL' when $i=6$, then let
 
 $E_1=subs(a_6=0,E_1)\backslash\{0\}=\{a_1 + a_2 - a_4 + a_9 + a_{10}- a_{11}-a_{12}, a_9 + a_{10} - a_{11}-a_{12}, - a_{9}- a_{10} + a_{11}+a_{12}\}$, 
 
 $S_{1}=S_{1}\backslash\{a_6\}=\{a_1,a_2,a_4,a_9,a_{10},a_{11},a_{12}\}$, 
 
 $F_1=subs(a_i=0,F_1)=-\frac{1}{2}(a_2-a_4-3a_9-a_{10}+a_{11}+a_{12})$.
 
 // For $i\in\mathcal{N}_{12}\backslash\{3,5,6,7,8\}$, Algorithm \ref{Heuristic search} outputs `UNSUCCESSFUL' and Algorithm \ref{FindComb} outputs `FALSE'.
 
 \textbf{Step 5}. Reduce $F_1$ by $E_1$ to obtain

\textbf{Problem \ref{pureproblem}$(*)$}. Prove $F_2$ subject to $\widetilde{R}(E_2)$ and $\mathcal{R}(V(\{F_2\}\cup E_2))$,

where $F_2=-\frac{1}{2}a_{2}+ \frac{1}{2}a_{4} - a_{10} + a_{11} + a_{12}$ and $E_2=\{a_1 + a_2 - a_4, a_9 + a_{10} - a_{11}-a_{12}\}$.
 \\

// Next, we will follow Algorithm \ref{Findredundant}.

\textbf{Step 6}. Use Algorithm \ref{QuickIMP-RED} to reduce Problem \ref{pureproblem}$(*)$ to 

\textbf{Problem \ref{problem-op1}$(*)$}. Proving $F_1\ge0$ subject to $\widetilde{\mathcal{R}}(E_1)$ and $\mathcal{R}(S_{1})$, 

where $F_1=\frac{1}{2}a_1-a_{10}+a_{11}+a_{12}$, $E_1=\{a_9 + a_{10} - a_{11}-a_{12}\}$, and $S_1=\{a_1,a_9,a_{10},a_{11},a_{12}\}$.

// Now we obtain Problem \ref{problem-op4} in Example \ref{III.5}.

\subsection{Formulas in \eqref{Final-1}}
\label{Appendix-B}
In this section, we list the formulas used in \eqref{Final-1}.
$$\begin{array}{ll}
	2h_{11} - h_{11, 12} =a_{6}, \\
	h_{6, 7, 8, 9, 10, 11, 12} - h_{2, 3, 8, 9, 10, 11, 12} =a_{85}, \\
	2h_{11, 12} - h_{3, 9, 10, 11, 12} - h_{11} =a_{94}, \\
	2h_{3, 8, 9, 11} - h_{2, 3, 6, 9, 10, 12} - h_{3, 9, 12} =a_{119}, \\
	h_{3, 9} + h_{9, 12} - h_{3, 8, 9} - h_{11} =a_{167}, \\
	h_{3, 9} + h_{11, 12} - h_{3, 8, 9} - h_{11} =a_{169}, \\
	h_{3, 8, 9} + h_{8, 9, 10, 12} - h_{3, 5, 7, 9} - h_{8, 10, 12} =a_{211}, \\
	h_{3, 9, 12} + h_{6, 7, 9, 10} - h_{1, 5, 10, 12} - h_{6, 9, 11} =a_{223}, \\
	h_{8, 11, 12} + h_{6, 9, 11} - h_{6, 7, 9, 10} - h_{9, 12} =a_{290}, \\
	h_{1, 5, 10, 12} + h_{6, 7, 9, 10, 11} - h_{3, 8, 9, 11, 12} - h_{6, 7, 9, 10} =a_{335}, \\
	h_{2, 3, 11, 12} + h_{3, 8, 10, 11} - h_{2, 3, 8, 10, 11} - h_{3, 8, 9, 11} =a_{340}, \\
	h_{2, 3, 11, 12} + h_{3, 5, 6, 7, 10, 12} - h_{2, 3, 8, 11, 12} - h_{3, 5, 6, 7, 9, 10} =a_{353}, \\
	h_{3, 8, 10, 12} + h_{3, 5, 7, 9} - h_{2, 3, 11, 12} - h_{8, 9, 10, 12} =a_{450}, \\
	h_{3, 9, 11, 12} + h_{8, 10, 12} - h_{3, 8, 10, 12} - h_{11, 12} =a_{484}, \\
	h_{6, 7, 9, 10} + h_{8, 9, 11, 12} - h_{6, 7, 9, 10, 11} - h_{8, 11, 12} =a_{519}, \\
	h_{2, 3, 8, 11, 12} + h_{3, 5, 7, 9, 10, 11, 12} - h_{2, 3, 8, 9, 10, 11, 12} - h_{3, 5, 6, 7, 10, 12} =a_{667}, \\
	h_{3, 8, 9, 11, 12} + h_{6, 8, 9, 11, 12} - h_{3, 5, 6, 7, 10, 12} - h_{8, 9, 11, 12} =a_{727}, \\
	h_{3, 9, 10, 11, 12} + h_{3, 8, 10, 12} - h_{8, 9, 10, 11, 12} - h_{3, 9, 11, 12} =a_{735}, \\
	h_{8, 9, 10, 11, 12} + h_{3, 5, 6, 7, 10, 12} - h_{3, 5, 7, 9, 10, 11, 12} - h_{3, 8, 10, 12} =a_{819}, \\
	h_{8, 9, 10, 11, 12} + h_{3, 5, 6, 7, 10, 12} - h_{3, 5, 7, 9, 10, 11, 12} - h_{3, 8, 9, 11, 12} =a_{820}, \\
	h_{8, 9, 10, 11, 12} + h_{2, 3, 8, 9, 10, 11, 12} - h_{6, 7, 8, 9, 10, 11, 12} - h_{3, 8, 10, 11} =a_{827}, \\
	h_{2, 3, 6, 9, 10, 12} + h_{2, 3, 8, 10, 11} - h_{2, 3, 8, 9, 10, 11, 12} - h_{2, 3, 11, 12} =a_{829}, \\
	h_{3, 5, 6, 7, 9, 10} + h_{3, 8, 9, 11, 12} - h_{3, 5, 6, 7, 10, 12} - h_{3, 8, 9, 11} =a_{859}, \\
	h_{3, 5, 6, 7, 10, 12} + h_{8, 9, 10, 11, 12} - h_{2, 3, 8, 9, 10, 11, 12} - h_{6, 8, 9, 11, 12} =a_{868}, \\
	h_{2, 3, 8, 9, 10, 11, 12} + h_{2, 3, 11, 12} - h_{2, 3, 5, 6, 7, 8, 10} - h_{3, 8, 10, 12} =a_{906}, \\
	h_{3, 5, 7, 9, 10, 11, 12} + h_{2, 3, 8, 9, 10, 11, 12} - h_{6, 7, 8, 9, 10, 11, 12} - h_{3, 5, 6, 7, 10, 12} =a_{916},\\
	\alpha - h_{3, 9}=a_{10188},\ \ 
	\beta - h_{11}=a_{10189}.
\end{array}$$

\end{document}